\documentclass[11pt]{amsart}
\usepackage{graphicx}
\usepackage{amscd}
\usepackage{amsmath}
\usepackage{amsfonts}
\usepackage{amssymb}
\usepackage{setspace}
\usepackage{enumerate}         % better lists
\usepackage{color}
\usepackage{url}
\usepackage{amsthm}
\usepackage{hyperref}
\usepackage{bm}
\usepackage{xy}
\usepackage{color}
\usepackage{subfig}
\allowdisplaybreaks[4]

\usepackage{geometry}
\theoremstyle{plain}
\newtheorem{theorem}{Theorem}[section]

\newtheorem{corollary}[theorem]{Corollary}

\newtheorem{lemma}[theorem]{Lemma}

\newtheorem{proposition}[theorem]{Proposition}

\newtheorem{definition}[theorem]{Definition}

\newtheorem{assumption}[theorem]{Assumption}

\theoremstyle{remark}
\newtheorem{remark}[theorem]{Remark}

\numberwithin{equation}{section}

\newcommand{\ind}{\mathbf{1}}
\newcommand{\rsto}{]\!\kern-1.8pt ]}
\newcommand{\lsto}{[\!\kern-1.7pt [}

\numberwithin{equation}{section}

\newcommand{\Lois}{L^{\text{OIS}}}

\newcommand{\RR}{\mathbb{R}}
\newcommand{\QQ}{\mathbb{Q}}

\newcommand{\CC}{\mathbb{C}}
\newcommand{\NN}{\mathbb{N}}
\newcommand{\EE}{\mathbb{E}}
\newcommand{\T}{\mathbb{T}}

\newcommand{\cA}{\mathcal{A}}
\newcommand{\cB}{\mathcal{B}}

\newcommand{\cF}{\mathcal{F}}
\newcommand{\cG}{\mathcal{G}}

\newcommand{\cI}{\mathcal{I}}
\newcommand{\cN}{\mathcal{N}}

\newcommand{\cP}{\mathcal{P}}

\newcommand{\cT}{\mathcal{T}}

\newcommand{\mfU}{\mathfrak{U}}
\newcommand{\mfD}{\mathfrak{D}}

\renewcommand{\Re}{\mathrm{Re}}
\renewcommand{\Im}{\mathrm{Im}}

\newcommand{\Excond}[3]{\mathbb{E}^{#1}\left[\left.#2\right|#3\right]}  %conditional expectation with superscript

\newcommand{\tr}{\mathop{\mathrm{Tr}}}
\newcommand{\im}{\ensuremath{\mathsf{i}}}

\makeatletter
\newcommand{\subjclassname@JEL}{JEL Classification}
\makeatother

\begin{document}
\title[Affine multiple yield curve models]{Affine multiple yield curve models}

\author{Christa Cuchiero}
\address[Christa Cuchiero]{University of Vienna, Faculty of Mathematics, \newline
\indent Oskar-Morgenstern Platz 1, 1090 Wien, Austria}
\email[Christa Cuchiero]{christa.cuchiero@univie.ac.at}%

\author{Claudio Fontana}
\address[Claudio Fontana]{Laboratoire de Probabilit\'es et Mod\`eles Al\'eatoires, Universit\'{e} Paris Diderot, \newline
\indent avenue de France, 75205 Paris, France}
\email[Claudio Fontana]{fontana@math.univ-paris-diderot.fr}

\author{Alessandro Gnoatto}
\address[Alessandro Gnoatto]{Mathematisches Institut der LMU M\"unchen,\newline%
\indent Theresienstrasse, 39 D-80333 M\"unchen}
\email[Alessandro Gnoatto]{alessandro@alessandrognoatto.com}%

\begin{abstract}
We provide a general and tractable framework under which all multiple yield curve modeling approaches based on affine processes, be it short rate, Libor market, or HJM modeling, can be consolidated. 
We model a num\'eraire process and multiplicative spreads between Libor rates and simply compounded OIS rates as functions of an underlying affine process. Besides allowing for ordered spreads and an exact fit to the initially observed term structures, this general framework leads to tractable valuation formulas for caplets and swaptions and embeds all existing multi-curve affine models. The proposed approach also gives rise to new developments, such as a short rate type model driven by a Wishart process, for which we derive a closed-form pricing formula for caplets. 
The empirical performance of two specifications of our framework is illustrated by calibration to market data.
\end{abstract}

\thanks{{\em Acknowledgements.} The authors are grateful to two anonymous referees for their valuable comments that helped to significantly improve the paper.}

\maketitle
%\tableofcontents

\section{Introduction}

Starting from the 2007 crisis, one of the most striking features of today's financial environment is represented by the segmentation of interest rate  markets. 
More specifically, while the pre-crisis interest rate market was well described by a \emph{single} yield curve and interbank Xibor rates\footnote{Similarly as in~\cite{bia10,GMS}, we denote by Xibor a generic interbank offered rate 
for unsecured term lending, such as the Libor rate in the London interbank market and the Euribor rate in the Eurozone. 
While the theoretical framework developed in the present paper applies to generic Xibor rates, the empirical results reported in Section 6 refer to  Euribor rates.} 
associated to different tenors were simply determined by no-arbitrage relations, nowadays the market is segmented in the sense that distinct yield curves are 
constructed from market instruments that depend on a specific tenor, thus giving rise to \emph{multiple} yield curves. The credit and liquidity risks existing in the interbank market, which were perceived as negligible before the crisis, are at the origin of this phenomenon. This is also reflected by the emergence of spreads between OIS rates and Xibor rates as well as between Xibor rates associated to different tenors. 

In this paper, we propose  
the most general tractable approach based on affine processes to model multiple yield curves. 
In our view, it should be understood as \emph{the affine framework} for this modeling purpose and not as a specific model class, such as the affine Libor models considered in \cite{GPSS14}. 
We provide a unifying affine methodology within which all interest rate modeling approaches, i.e., (i) short rate, (ii) Libor market, and (iii) Heath-Jarrow-Morton modeling, can be consolidated.
On the one hand, this allows us to embed in our framework all affine multiple curve models proposed so far in the literature. 
On the other hand, the approach allows for novel developments, for instance multiple curve short rate models based on Wishart processes as well as interest rate models under alternative num\'eraires. 
Moreover, the generality of the setting does not preclude tractability and  appealing practical features, as will be made precise below.

We model a general num\'eraire process and, inspired by our previous work \cite{CFG:14}, \emph{multiplicative spreads} between (normalized) 
spot Xibor rates and (normalized) simply compounded OIS rates. Besides being directly observable from market quotes, 
multiplicative spreads represent the most convenient modeling quantity in an affine setting and admit a natural interpretation in terms of forward exchange premia. 
We consider a general interest rate market where OIS zero-coupon bonds and FRA contracts are traded, for a finite set of tenors and for all maturities, 
and assume the existence of a num\'eraire - martingale measure couple, thus ensuring absence of arbitrage in a sense precisely specified below. 
We then model the logarithm of the multiplicative spreads and of the num\'eraire process as affine functions of an underlying affine process. 
In the special case where the num\'eraire is chosen as the OIS bank account, our framework can be regarded as the natural extension of 
classical affine short rate models to the multi-curve setting. 
In the special case where the num\'eraire is chosen to be the OIS bond with terminal maturity, we can recover an extension of the affine Libor market model of \cite{GPSS14}.

The proposed framework fully exploits the analytical tractability of affine processes and exhibits several desirable modeling features. Indeed, besides its generality, it allows for
\begin{itemize}
\item
an automatic fit to the initially observed term structures of OIS bond prices and spreads;
\item a transparent characterization of order relations between spreads associated to different tenors;
\item efficient valuation formulae for caplets and swaptions in a multi-curve setting. 
\end{itemize}
The driving affine process is allowed to take values in a general convex state space. In particular, we analyze a specification based on Wishart processes and we derive a novel closed-form formula for the price of a caplet, expressed in terms of the distribution function of weighted sums of independent non-central chi-square-distributed random variables. 
Beyond that, we test the empirical performance of two simple specifications of our setup and show that they achieve a satisfactory fit to market data.

Referring to the introductory section of \cite{CFG:14} and to the monographies \cite{GR15,Henr14} for a more detailed overview of the relevant literature, we mention here that affine processes have been used to model multiple yield curves in \cite{fitr12,GM:14,GMR:15,ken10,kitawo09,MR14} by adopting a short rate approach and, more recently, in \cite{GPSS14} by extending the affine Libor model originally introduced in \cite{KRPT13}. 
As already mentioned, the flexibility of our approach is highlighted by the fact that, to the best of our knowledge, all existing affine multi-curve models can be recovered as special cases of our framework. For this reason, the present work structurally differs from papers on specific affine modeling approaches, as previously considered in the literature. Beyond that, it also provides a broad class of tractable specifications of the abstract HJM approach based on general It\^o semimartingales  proposed in \cite{CFG:14}  for arbitrage-free modeling of the term structures of OIS bonds and multiplicative spreads.

The paper is organized as follows. 
Section~\ref{sec:intro2} starts by introducing the main market rates and multiplicative spreads and then presents an abstract approach to the arbitrage-free modeling of a general multi-curve interest rate market.
In Section~\ref{affinespecification}, we develop a general framework driven by affine processes as well as possible specifications based on a short rate approach. Moreover, we  relate our setup to several affine multi-curve models recently proposed in the literature.
In Section~\ref{sec:aff_pricing}, we derive general semi-closed valuation formulae for caplets and an approximate formula for the price of a swaption.
Section~\ref{sec:Wishart} contains a detailed analysis of a tractable specification based on Wishart processes, deriving an analytical pricing formula for caplets.
Section~\ref{sec:aff_examples} presents the calibration results for two specifications of the framework.
Finally,  Appendix~\ref{sec:gen_pricing} presents general pricing formulae in terms of the quantities used in our framework and Appendix~\ref{app:proof_prop} contains the proof of Proposition~\ref{prop:bonds_spreads}.

\section{Modeling the post-crisis interest rate market}	\label{sec:intro2}

\subsection{Xibor rates, OIS rates and spot multiplicative spreads} 	\label{sec:rates}

The underlying quantities of most interest rate products are represented by Xibor rates. We denote by $L_t(t,t+\delta)$ the Xibor rate prevailing at date $t$ for the time interval $[t,t+\delta]$, where the tenor $\delta>0$ is typically one day (1D), one week (1W), or several months (1M, 2M, 3M, 6M or 12M). Referring to \cite[Chapter 1]{GR15} for a more detailed presentation of post-crisis interest rate markets, we just mention that Xibor rates are reference rates determined by a panel of primary financial institutions for unsecured lending and are not based on actual transactions.
In the post-crisis fixed income market, Xibor rates associated to different tenors started to exhibit distinct behaviors, leading for instance to non-negligible basis spreads (see Appendix A.1). In this paper, we shall consider Xibor rates for a generic set of tenors $\{\delta_1,\ldots,\delta_m\}$, with $\delta_1<\ldots<\delta_m$, for some $m\in\mathbb{N}$. 

The reference rates for overnight borrowing (i.e., for the shortest tenor of 1D) correspond to the Eonia (Euro overnight index average) rate in the Eurozone and to the Federal Funds rate in the US market. 
Unlike Xibor rates, the Eonia and the Federal Funds rates are determined on the basis of actual overnight transactions in the interbank market by a panel of banks (in the case of the Eonia rate, this is the same panel determining the Euribor rate)\footnote{Note also that, besides being based on actual interbank unsecured transactions, the Eonia and the Federal Funds rate have also a different settlement period than 1D Xibor rates. We refer to \cite{bia_slides} for a detailed presentation of the different market conventions of Libor/Euribor rates and Eonia rates.}.
Overnight rates represent the underlying of overnight indexed swaps (OIS), whose market swap rates are referred to as OIS rates (see Appendix~\ref{sec:noopt}), typically considered as the best proxy for risk-free rates in market practice.
By bootstrapping techniques (see~\cite{AB:13}), OIS rates allow to recover the OIS term structure $T\mapsto B(t,T)$, where $B(t,T)$ denotes the price at date $t$ of an OIS zero-coupon bond with maturity $T$. We define the simply compounded (risk-free) OIS spot rate by
\[
\Lois_t(t,t+\delta) := \frac{1}{\delta}\left(\frac{1}{B(t,t+\delta)}-1\right),
\]
for $\delta>0$. 
In particular, note that the right-hand side of the above formula corresponds to the pre-crisis textbook definition of Xibor rate.

Our main modeling quantities are the following \emph{spot multiplicative  spreads}:
\begin{equation}	\label{eq:multspread}
S^{\delta_i}(t,t):=\frac{1+\delta_i L_t(t,t+\delta_i)}{1+\delta_i L^{{\rm OIS}}_t(t,t+\delta_i)},
\qquad \text{ for }i=1,\ldots,m,
\end{equation}
corresponding to multiplicative spreads between (normalized) Xibor rates and (normalized) simply compounded OIS spot rates. The multiplicative spreads $S^{\delta_i}(t,t)$ can be directly inferred from the quoted Xibor and OIS rates. Indeed, the numerator of \eqref{eq:multspread} is determined by the Xibor rate quoted on the market, while the quantity appearing in the denominator can be bootstrapped from quoted OIS rates, as mentioned above.
In the post-crisis environment, the quantities $S^{\delta_i}(t,t)$ are usually greater than one and increasing with respect to the tenor's length $\delta_i$. Neglecting liquidity issues, this is related to the fact that Xibor rates embed the risk that the average credit quality of the banks included in the Xibor panel deteriorates over the term of the loan, while OIS rates reflect the average credit quality of a newly refreshed Xibor panel (see, e.g.,~\cite{CDS:01, fitr12}). As will be shown below, a key feature of our modeling framework is the facility of generating multiplicative spreads satisfying such requirements. 

In comparison to additive spreads (as considered for instance in \cite{mer10,merxie12}), multiplicative spreads admit a natural economic interpretation in a multiple curve setting and, as shown in the following sections, represent a convenient modeling quantity  in relation with affine processes. Referring to \cite[Appendix B]{CFG:14} for full details, Xibor rates can be associated with artificial risky bond prices $B^{\delta_i}(t,T)$, so that $1+\delta_i L_t(t,t+\delta_i)=1/B^{\delta_i}(t,t+\delta_i)$, for $i=1,\ldots,m$. If risky bonds are interpreted as bonds of a foreign economy, with $B^{\delta_i}(t,T)$ representing the price (in units of the foreign currency) of a foreign zero-coupon bond, and OIS bonds are interpreted as domestic bonds, then the quantity $S^{\delta_i}(t,t)$ corresponds to the {\em forward exchange premium} between the domestic and the foreign economy over the period $[t,t+\delta_i]$. In this sense, the multiplicative spread $S^{\delta_i}(t,t)$ can be regarded as a market expectation (at date $t$) of the riskiness of the Xibor panel over the period $[t,t+\delta_i]$. Related foreign exchange analogies have been proposed in \cite{bia10} and \cite{NS15}.

\subsection{Basic traded assets and fundamental properties}	\label{sec:abstract_setting}

Let us now make precise the financial market considered in this paper. Among all financial contracts written on Xibor rates, forward rate agreements (FRA) can be rightfully considered as the basic building blocks, due to the simplicity of their payoff and to the fact that all linear interest rate derivatives (such as interest rate swaps and basis swaps) can be represented as portfolios of FRAs. 
As usually done in the literature, we consider standard ``textbook'' FRAs, in contrast to the cash-settled FRAs traded in the market (compare with~\cite[Remark 1.3]{GR15}). While the payoff of market FRAs is non-linear, the discrepancy between the two versions of FRAs is typically negligible for practical purposes, as documented in~\cite{mer10}.

Motivated by the above observation, we shall consider a general financial market consisting of the following families of \emph{basic traded assets}, with $\T<+\infty$ denoting a fixed time horizon:
\begin{itemize}
\item[(i)] OIS zero-coupon bonds for all maturities $T\in[0,\T]$;
\item[(ii)] FRA contracts for all maturities $T\in[0,\T]$ and for all tenors $\{\delta_1,\ldots,\delta_m\}$.
\end{itemize}
Note that, differently from the pre-crisis single-curve setting, FRA contracts have to be considered on top of OIS bonds since they cannot be perfectly replicated by the latter any longer, due to the discrepancy between Xibor rates $L_t(t,t+\delta_i)$ and simply compounded OIS rates $\Lois_t(t,t+\delta_i)$.
Moreover, OIS bonds play a particularly important role in collateralized transactions. Indeed, the collateral rate adopted in collateralized transactions in interest rate markets is typically specified as the OIS short rate and, hence, OIS bonds represent the natural discount factors for the valuation of interest rate products (see, e.g.,~\cite{fitr12} and \cite[Appendix A]{CFG:14}).
We denote by $\Pi^{FRA}(t;T,T+\delta_i,K)$ the price at date $t$ of a FRA contract starting at date $T$ with maturity $T+\delta_i$, rate $K>0$ and unitary notional, for $0\leq t\leq T\leq \T$ and $i=1,\ldots,m$. We denote by $L_t(T,T+\delta_i)$ the corresponding \emph{fair FRA rate}, i.e., the rate $K$ fixed at date $t$ satisfying $\Pi^{FRA}(t;T,T+\delta_i,K)=0$. We say that a FRA contract written at $t$ is \emph{fair} if its rate $K$ is set equal to $L_t(T,T+\delta_i)$, i.e., if the price of the FRA contract at the inception date $t$ is zero.
Extending definition \eqref{eq:multspread} of the spot multiplicative spread, we define (forward) \emph{multiplicative spreads} $S^{\delta_i}(t,T)$, for $0\leq t\leq T\leq \T$ and $i=1,\ldots,m$, by
\begin{equation}	\label{eq:spreadtT}
S^{\delta_i}(t,T) := \frac{1+\delta_i L_t(T,T+\delta_i)}{1+\delta_i\Lois_t(T,T+\delta_i)} = \bigl(1+\delta_i L_t(T,T+\delta_i)\bigr)\frac{B(t,T+\delta_i)}{B(t,T)}.
\end{equation}
Multiplicative spreads of the form \eqref{eq:spreadtT} have been first introduced in \cite{hen10} and further considered as main modeling quantities in \cite{CFG:14} (see also \cite{NS15}).
We refer to \cite[Appendix B]{CFG:14} for the interpretation of the forward multiplicative spread $S^{\delta_i}(t,T)$ in terms of the foreign exchange analogy introduced at the end of Section \ref{sec:rates}.

A first and fundamental issue is the absence of arbitrage in the financial market composed by the basic traded assets introduced above. To this end, we let $(\Omega,\cF)$ be a measurable space endowed with a right-continuous filtration $(\cF_t)_{0\leq t\leq \T}$ and formulate the following definition.
 
\begin{definition}	\label{def:couple}
Let $\QQ$ be a probability measure on $(\Omega,\cF)$ and $B=(B_t)_{0\leq t\leq \T}$ a strictly positive adapted process with $B_0=1$. We say that $(B,\QQ)$ is a \emph{num\'eraire - martingale measure couple} if the $B$-discounted price of every basic traded asset is a martingale on $(\Omega,(\cF_t)_{0\leq t\leq \T},\QQ)$.
\end{definition}

In other words, if $(B,\QQ)$ is a num\'eraire - martingale measure couple, the price of every OIS zero-coupon bond and FRA contract is a $\QQ$-martingale when discounted with respect to the num\'eraire process $B$. Definition~\ref{def:couple} is in the spirit of~\cite{DS:01}, where the authors consider term structure models generated by a num\'eraire - martingale measure couple in the above sense. Note that we do not assume that $B$ is a traded asset and it can therefore be interpreted as a state-price density process when $\mathbb{Q}$ corresponds to the physical probability measure (compare with~\cite{FH:96, R:97}).

In the present paper, we shall work under the following standing assumption.

\begin{assumption}	\label{ass:couple}
There exists a num\'eraire - martingale measure couple $(B,\QQ)$.
\end{assumption}

In particular, Assumption \ref{ass:couple} suffices to exclude the existence of arbitrage profits, in the sense of \emph{no asymptotic free lunch with vanishing risk} (NAFLVR, see \cite{CKT14}), in the large financial market composed by all basic traded assets, i.e., by all OIS zero-coupon bonds and FRA contracts. 

\begin{remark}[OIS bank account as num\'eraire]	\label{rem:OIS_num}
In the literature on multiple curves, the num\'eraire is typically chosen to be the OIS bank account and $\QQ$ is the corresponding (spot) martingale measure. While this choice of $(B,\QQ)$ will be considered in detail in Section \ref{sec:short_rate}, we want to point out that our framework is not limited to this specification and allows to consider several alternative settings, depending on the choice of the num\'eraire - martingale measure couple (see Section \ref{sec:aff_model} below). For the moment, we just think of the couple $(B,\QQ)$ in abstract and general terms.
\end{remark}

Assumption \ref{ass:couple} implies that the arbitrage-free price of every basic traded asset can be computed by taking the conditional $\QQ$-expectation of its $B$-discounted payoff, similarly as in the classical risk-neutral valuation paradigm\footnote{In the typical situation where $B$ is the OIS bank account and $\QQ$ the corresponding martingale measure, as considered in Remark \ref{rem:OIS_num}, computing the price of a derivative by taking the conditional $\QQ$-expectation of its $B$-discounted payoff corresponds to computing its \emph{clean price}, i.e., neglecting counterparty risk and assuming a funding rate equal to the OIS rate (compare also with the discussion in~\cite[Appendix A]{CFG:14}).}. Throughout the paper, unless explicitly indicated, expectations are taken under the measure $\QQ$, while the measure $\QQ^T$ denotes the $T$-forward measure with density $d\QQ^T/d\QQ=1/(B_TB(0,T))$, for $T\in[0,\T]$. We can now state the following proposition, which is a simple consequence of Definition~\ref{def:couple} (for completeness, the proof is given in Appendix~\ref{app:proof_prop}). 
 
\begin{proposition}	\label{prop:bonds_spreads}
Suppose that Assumption \ref{ass:couple} holds. Then the following hold:
\begin{enumerate}
\item
OIS zero-coupon bond prices satisfy
\begin{equation}	\label{eq:bond_gen}
B(t,T) = \EE\left[\frac{B_t}{B_T}\Bigr|\cF_t\right],
\qquad \text{ for all }0\leq t\leq T\leq \T;
\end{equation}
\item 
for every $i=1,\ldots,m$, the fair FRA rate satisfies
\begin{equation}	\label{eq:libor_gen}
L_t(T,T+\delta_i) = \EE^{\QQ^{T+\delta_i}}[L_T(T,T+\delta_i)|\cF_t],
\qquad \text{ for all }0\leq t\leq T\leq \T;
\end{equation}
\item
for every $i=1,\ldots,m$, the multiplicative spread satisfies
\begin{equation}	\label{eq:spread_gen}
S^{\delta_i}(t,T) = \EE^{\QQ^T}[S^{\delta_i}(T,T)|\cF_t],
\qquad \text{ for all }0\leq t\leq T\leq \T.
\end{equation}
\end{enumerate}
\end{proposition}

From the modeling perspective, Proposition \ref{prop:bonds_spreads} together with representation \eqref{eq:spreadtT} implies that, in order to give an arbitrage-free description of the financial market composed by all OIS zero-coupon bonds and FRA contracts, it suffices to model the num\'eraire process $B$ together with the spot multiplicative spreads $\{(S^{\delta_i}(t,t))_{0\leq t\leq \T}, i=1,\ldots,m\}$, under some reference measure $\QQ$. 
In the next section, this will be achieved in a flexible and tractable way by letting the num\'eraire process and the spot multiplicative spreads be driven by a common underlying affine process.

\section{Multi-curve models based on affine processes}\label{affinespecification}

In this section we develop a general framework based on affine processes for modeling multiple curves. We first recall in Section~\ref{sec:aff_gen} some general results on affine processes, mainly relying on \cite{cfmt11,CT:13,krm12}. In Section~\ref{sec:aff_model}, we provide the general definition of an affine multi-curve model with respect to an abstract num\'eraire-martingale measure couple and derive its fundamental properties. This general setup is then specialized in Section~\ref{sec:short_rate} to the typical setting where the num\'eraire is given by the OIS bank account modeled via a short rate. Section~\ref{sec:relations} shows how our setup relates to the existing multi-curve models driven by affine processes.

\subsection{Preliminaries on affine processes}	\label{sec:aff_gen}

Let $V$ be a finite-dimensional real vector space with associated scalar product $\langle \cdot,\cdot\rangle$. We denote by $D$ a non-empty closed  convex subset of $V$, which will serve as state space for the stochastic process to be introduced below, endowed with the Borel $\sigma$-algebra $\cB_D$. 
Typical choices for the state space are represented by $D\in\{\RR^d,\RR_+^{d},\mathbb{S}^{d}_+\}$ or combinations thereof, for some $d\in\NN$, where $\mathbb{S}^{d}_+$ denotes the cone of symmetric ${d}\times {d}$ positive semidefinite matrices.\footnote{If $D=\mathbb{S}^d_+$ and $x,y\in\mathbb{S}_+^d$, the scalar product is given by $\langle x,y\rangle=\tr[xy]$, with $\tr[\cdot]$ being the trace operator.} 

On a filtered probability space $(\Omega,\cF,(\cF_t)_{0\leq t\leq \mathbb{T}},\QQ)$ satisfying the usual conditions, we introduce a stochastic process $X=\left(X_t\right)_{0\leq t\leq \T}$, starting from an initial value $X_0=x$ belonging to the interior of $D$. 
The process $X$ is assumed to be an adapted, time-homogeneous and conservative Markov process taking values in $D$ and we denote by $\{p_t:D\times\cB_D\rightarrow[0,1];t\in[0,\T]\}$ the family of its transition kernels.
To the process $X$ we associate the sets
\begin{align}	\label{eq:expmom}
\mfU_T := \bigl\{\zeta\in V+\im V : 
\mathbb{E}\bigl[e^{\langle \zeta,X_t\rangle}\bigr]<+\infty, 
\text{ for all }t\in[0,T]\bigr\}
\qquad\text{ and }\qquad
\mfU:=\mfU_{\mathbb{T}}.
\end{align}
Letting $\mfD:=\{(t,\zeta)\in[0,\T]\times(V+\im V) : \zeta\in\mfU_t\}$, we define as follows the class of affine processes.

\begin{definition} 	\label{def:affine}
The Markov process $X=(X_t)_{0\leq t\leq \T}$ is said to be \emph{affine} if
\begin{enumerate}
\item it is stochastically continuous, i.e., the transition kernels satisfy $\lim_{s\to t}p_s(x,\cdot)=p_t(x,\cdot)$ weakly on $D$, for every $(t,x)\in[0,\T]\times D$;
\item its Fourier-Laplace transform has exponential affine dependence on the initial state, i.e., there exist functions $\phi:\mfD\rightarrow \CC$ and $\psi:\mfD\rightarrow V+\im V$ such that, for every initial value $x\in D$ and for every $(t,u)\in\mfD$, it holds that
\begin{align}
\mathbb{E}[e^{\langle u,X_t\rangle}]
=\int_De^{\langle u,\xi\rangle}p_t(x,d\xi)=e^{\phi(t,u)+\langle\psi(t,u),x\rangle}.\label{affineproperty}
\end{align}
\end{enumerate}
\end{definition}

\begin{remark}	\label{rem:convex_expmom}
\begin{enumerate}
\item
Definition \ref{def:affine} slightly differs from the usual definition of affine process where the exponential affine form \eqref{affineproperty} is only assumed to hold on the set of bounded exponentials. Indeed, in most theoretical works the definition of affine process only refers to the characteristic function or Laplace transform (for processes with values in a cone). 
Then a natural question arises, namely if the exponential affine form can be extended beyond the characteristic function or Laplace transform, in our notation to $\mathcal{D}$. This important question has been answered positively in \cite{krm12} under general assumptions and, therefore, we modified the classical definition.
In particular, in view of~\cite[Lemma 4.2]{krm12}, it holds that 
\[
\mfU_T=\bigl\{\zeta\in V+\im V : \, \mathbb{E}\bigl[e^{\langle \zeta,X_T\rangle}\bigr]<+\infty\bigr\},
\qquad\text{ for all }T\in[0,\T].
\]
\item
At the expense of greater technicalities, our setup can be extended to affine processes taking values in general (possibly non-convex) state spaces. However, the class of affine processes considered in Definition \ref{def:affine} is sufficiently rich to cover any affine multi-curve model of practical interest.
\end{enumerate}
\end{remark}

The affine property \eqref{affineproperty} together with the Chapman-Kolmogorov equation implies that the functions $\phi$ and $\psi$ have the \textit{semiflow property}, i.e., for every $u\in \mfU$ and $t,s\in[0,\T]$ with $s+t\leq\T$, it holds that
\begin{equation}
\begin{split}
\phi(t+s,u)&=\phi(t,u)+\phi\bigl(s,\psi(t,u)\bigr),\\
\psi(t+s,u)&=\psi\bigl(s,\psi(t,u)\bigr).\label{semiflow}
\end{split}
\end{equation}
The stochastic continuity of the affine process $X$ also implies its  \emph{regularity}, according to \cite[Theorem 7]{CT:13} (see also \cite[Theorem 4.3]{krst11} in the case $D=\RR^n_+\times\RR^{d-n}$, for some $n\in\{0,1,\ldots,d\}$, and \cite[Proposition 3.4]{cfmt11} in the case $D=\mathbb{S}^d_+$), in the sense that the derivatives
\begin{equation}	\label{eq:F_R}
F(u):=\left.\frac{\partial \phi(t,u)}{\partial t}\right|_{t=0+}
\quad\text{ and }\quad 
R(u):=\left.\frac{\partial \psi(t,u)}{\partial t}\right|_{t=0+}
\end{equation}
exist and are continuous at $u=0$. Regularity implies that it is possible to differentiate the semiflow relations \eqref{semiflow} with respect to $t$ and evaluate them at $t=0$, thus obtaining the following system of generalized Riccati ODEs
\begin{align*}
\frac{\partial \phi(t,u)}{\partial t}&=F\bigl(\psi(t,u)\bigr),\quad \phi(0,u)=0;\\
\frac{\partial \psi(t,u)}{\partial t}&=R\bigl(\psi(t,u)\bigr),\quad \psi(0,u)=u\in\mfU.
\end{align*}

We state the following result (see~\cite[Theorem 4.10]{phdthesis_kr}), which will be useful for some specifications of our setup. We refer to~\cite{K:06} for the notion of differential characteristics of an affine process and we denote by $V^1$ and $V^2$ two real vector spaces. 

\begin{lemma}\label{mixedstate}
Let $X=(X^1,X^2)$ be an affine process taking values in a state space of the form $D=D^1\times D^2\subseteq V^1\times V^2$, with initial value $X_0=(x^1,x^2)$ and such that the differential characteristics of $X^2$ only depend on $X^1$.
Then there exist functions $\tilde{\phi}:\mfD\rightarrow\mathbb{C}$ and $\tilde{\psi}:\mfD\rightarrow V^1+\im V^1$ such that, for every $(t,u_1,u_2)\in\mfD$, it holds that
\[
\mathbb{E}\left[e^{\langle u_1,X^1_t\rangle+\langle u_2,X^2_t\rangle}\right]=e^{\tilde{\phi}(t,u_1,u_2)+\langle \tilde{\psi}(t,u_1,u_2),x^1\rangle+\langle u_2,x^2\rangle}.
\]
\end{lemma}
Typical examples where affine processes satisfying the assumptions of Lemma \ref{mixedstate} are employed are short rate models (see Section \ref{sec:short_rate}) and affine stochastic volatility models (see, e.g.,~\cite[Chapter 5]{phdthesis_Cuchiero}).

\subsection{Definition and general properties of affine multi-curve models}	\label{sec:aff_model}

As explained in Section~\ref{sec:abstract_setting}, our general setup assumes the existence of a num\'eraire-martingale measure couple $(B,\QQ)$, meaning that the $B$-discounted price of every traded asset is  a $\QQ$-martingale. Recall also that we consider as basic traded assets the OIS zero-coupon bonds and the FRA contracts, for all maturities $T\in[0,\T]$ and for a finite set of tenors $\{\delta_1,\ldots,\delta_m\}$, with $\delta_1<\ldots<\delta_m$, for some $m\in\mathbb{N}$.

Letting $X=(X_t)_{0\leq t\leq \T}$ be an affine process taking values in a state space $D\subseteq V$ under a measure $\QQ$, we are now in a position to give the general definition of an \emph{affine multi-curve model}. 

\begin{definition}	\label{def:aff_model}
Let $\mathbf{u}=(u_0,u_1,\ldots,u_m)$ be a family of functions $u_i:[0,\T]\rightarrow V$, $i=0,1,\ldots,m$, such that $u_0(t)\in\mfU_t$ and $u_i(t)+u_0(t)\in\mfU_t$, for every $i=1,\ldots,m$ and $t\in[0,\T]$, and $\mathbf{v}=(v_0,v_1,\ldots,v_m)$ a family of functions $v_i:[0,\T]\rightarrow\RR$, $i=0,1,\ldots,m$.
We say that the triplet $(X,\mathbf{u},\mathbf{v})$ is an \emph{affine multi-curve model} if
\begin{enumerate}
\item 
the logarithm of the num\'eraire process $B_t$ satisfies
\begin{equation}	\label{eq:log_numeraire}
\log B_t = -v_0(t)-\langle u_0(t),X_t\rangle,
\qquad\text{ for all } t\in[0,\T];
\end{equation}
\item
the logarithm of the spot multiplicative spreads $\{S^{\delta_i}(t,t);i=1,\ldots,m\}$ satisfies
\begin{equation}	\label{eq:log_spread}
\log S^{\delta_i}(t,t) = v_i(t)+\langle u_i(t),X_t\rangle,
\qquad\text{ for all } t\in[0,\T] \text{ and }i=1,\ldots,m.
\end{equation}
\end{enumerate}
\end{definition}

Depending on the choice of the num\'eraire-martingale measure couple $(B,\QQ)$ and of the triplet $(X,\mathbf{u},\mathbf{v})$, different modeling approaches can be obtained from the present general setup. For instance:
\begin{itemize}
\item if the num\'eraire asset $B$ is chosen to be the OIS bank account, $\QQ$ is the corresponding (spot) martingale measure and $B$ is modeled via an OIS short rate $r=(r_t)_{0\leq t\leq \T}$, then Definition \ref{def:aff_model} embeds an affine framework for modeling the OIS short rate and the multiplicative spreads. This specification will be analyzed in detail in Section \ref{sec:short_rate} (compare also with Section \ref{sec:rel_short_rate});
\item if the OIS zero-coupon bond with maturity $\T$ is chosen as num\'eraire and $\QQ$ is the $\T$-forward measure, then Definition \ref{def:aff_model} yields an extension of the affine Libor model with multiple curves recently introduced in \cite{GPSS14} (see Section \ref{sec:rel_affineLibor});
\item if $\QQ$ is the physical probability measure and $B$ is chosen as the \emph{growth-optimal portfolio} (in the spirit of the
\emph{benchmark approach} to interest rate modeling, see, e.g.,~\cite{BLNSP10,TP15} and, for a general account of the benchmark approach, \cite{PH}) or, more generally, as a {\em state-price density} (as recently considered in the rational multi-curve models \cite{CMNS15,FLT:14, NS15}), then Definition \ref{def:aff_model} leads to a {\em real-world} approach to the modeling of multiple curves. 
\end{itemize}

Affine multi-curve models lead to a tractable structure for OIS bond prices and multiplicative spreads. More precisely, affine multi-curve models can be characterized via the exponentially affine form of OIS bond prices and spreads. To this end, let us give the following definition.

\begin{definition}	\label{def:QAffineModel}
The affine process $X$ is said to \emph{generate exponentially affine OIS bond prices and spreads} if
\begin{enumerate}
\item
the $B$-discounted prices of OIS zero-coupon bonds satisfy
\begin{equation} \label{eq:generalAffineBond}
\frac{B(t,T)}{B_t} = \exp\bigl(\cA^0(t,T)+\langle \cB^0(t,T),X_t\rangle\bigr),
\qquad\text{ for all }0\leq t\leq T\leq\T,
\end{equation}	
for some functions $\cA^0:[0,\T]\times[0,\T]\rightarrow\RR$ and $\cB^0:[0,\T]\times[0,\T]\rightarrow V$, and
\item 
the multiplicative spreads satisfy
\begin{equation}	\label{eq:generalAffineSpread}
S^{\delta_i}(t,T) = \exp\bigl(\cA^i(t,T)+\langle \cB^i(t,T),X_t\rangle\bigr),
\qquad\text{ for all }0\leq t\leq T\leq\T,
\end{equation}
for some functions $\cA^i:[0,\T]\times[0,\T]\rightarrow\RR$ and $\cB^i:[0,\T]\times[0,\T]\rightarrow V$, for all $i=1,\ldots,m$.
\end{enumerate}
\end{definition}

The following proposition provides the announced characterization of affine multi-curve models, thereby showing that Definition \ref{def:aff_model} is equivalent to Definition \ref{def:QAffineModel}.

\begin{proposition}	\label{prop:QAffineCharacterization}
Let $X$ be an affine process. Then the following hold:
\begin{enumerate}
\item if $(X,\mathbf{u},\mathbf{v})$ is an affine multi-curve model, then $X$ generates exponentially affine OIS bond prices and spreads and the functions $\cA^i$, $\cB^i$, $i=0,1,\ldots,m$, appearing in \eqref{eq:generalAffineBond}-\eqref{eq:generalAffineSpread} are given by
\begin{equation}	\label{eq:correspondence}	\begin{aligned}
\cA^0(t,T) &= v_0(T)+\phi\bigl(T-t,u_0(T)\bigr),\\
\cB^0(t,T) &= \psi\bigl(T-t,u_0(T)\bigr),\\
\cA^i(t,T) &= v_i(T)+\phi\bigl(T-t,u_i(T)+u_0(T)\bigr)-\phi\bigl(T-t,u_0(T)\bigr),\\
\cB^i(t,T) &= \psi\bigl(T-t,u_i(T)+u_0(T)\bigr)-\psi\bigl(T-t,u_0(T)\bigr),
\end{aligned}	\end{equation}
for all $0\leq t\leq T\leq \T$ and $i=1,\ldots,m$, where $\phi$ and $\psi$ denote the characteristic exponents of $X$ as in Definition \ref{def:affine};
\item 
conversely, if $X$ generates exponentially affine OIS bond prices and spreads, then $(X,\mathbf{u},\mathbf{v})$ is an affine multi-curve model, where the functions $u_i,v_i$, $i=0,1,\ldots,m$, are given by
\begin{equation}	\label{eq:correspondence_2}
v_0(t) = \cA^0(t,t),
\qquad u_0(t) = \cB^0(t,t),	
\qquad v_i(t) = \cA^i(t,t),
\qquad u_i(t) = \cB^i(t,t),
\end{equation}
for all $t\in[0,\T]$ and $i=1,\ldots,m$.
\end{enumerate}
\end{proposition}

\begin{proof}
Suppose first that $(X,\mathbf{u},\mathbf{v})$ is an affine multi-curve model. Then, in view of Proposition \ref{prop:bonds_spreads}, OIS bond prices are given by
\[
\frac{B(t,T)}{B_t} = \EE\left[\frac{1}{B_T}\Bigr|\cF_t\right]
= \EE\left[e^{v_0(T)+\langle u_0(T),X_T\rangle}\bigr|\cF_t\right]
= e^{v_0(T)+\phi(T-t,u_0(T))+\langle\psi(T-t,u_0(T)),X_t\rangle},
\]
where the last equality follows from \eqref{affineproperty} together with the assumption that $u_0(T)\in\mfU_T$ (see also Remark \ref{rem:convex_expmom}). This proves that representation \eqref{eq:generalAffineBond} holds with the functions $\cA^0$ and $\cB^0$ being given as in \eqref{eq:correspondence}. Similarly, Proposition \ref{prop:bonds_spreads} together with the assumption that $u_i(T)+u_0(T)\in\mfU_T$, for every $T\in[0,\T]$ and $i=1,\ldots,m$, implies that the multiplicative spreads satisfy
\begin{align*}
S^{\delta_i}(t,T) &= \EE^{\QQ^T}[S^{\delta_i}(T,T)\bigr|\cF_t]
= \EE^{\QQ^T}\left[e^{v_i(T)+\langle u_i(T),X_T\rangle}\bigr|\cF_t\right]	\\
&= \frac{B_t}{B(t,T)}\EE\left[e^{v_i(T)+v_0(T)+\langle u_i(T)+u_0(T),X_T\rangle}\bigr|\cF_t\right]	\\
&= e^{v_i(T)+\phi(T-t,u_i(T)+u_0(T))-\phi(T-t,u_0(T))+\langle\psi(T-t,u_i(T)+u_0(T))-\psi(T-t,u_0(T)),X_t\rangle},
\end{align*}
thus showing that representation \eqref{eq:generalAffineSpread} holds with the functions $\cA^i$ and $\cB^i$ being given as in \eqref{eq:correspondence}.

Conversely, suppose that $X$ generates exponentially affine OIS bond prices and spreads. Then, for every $t\in[0,\T]$, letting $T=t$ in \eqref{eq:generalAffineBond}-\eqref{eq:generalAffineSpread} and recalling that $B(t,t)=1$, it follows that
\begin{align*}
\log B_t = -\cA^0(t,t) - \langle\cB^0(t,t),X_t\rangle
\qquad\text{and}\qquad
\log S^{\delta_i}(t,t) = \cA^i(t,t) + \langle\cB^i(t,t),X_t\rangle.
\end{align*}
Defining the functions $u_i$ and $v_i$ as in \eqref{eq:correspondence_2}, this proves that representations \eqref{eq:log_numeraire}-\eqref{eq:log_spread} hold true.
It remains to prove that $u_0(T)=\cB^0(T,T)\in\mfU_T$ and $u_i(T)+u_0(T)=\cB^i(T,T)+\cB^0(T,T)\in\mfU_T$, for every $T\in[0,\T]$ and $i=1,\ldots,m$. This follows from Proposition \ref{prop:bonds_spreads}, since
\begin{align*}
\EE\left[e^{\langle u_0(T),X_T\rangle}\right] &= e^{-v_0(T)}\EE\left[\frac{1}{B_T}\right] = e^{-v_0(T)}B(0,T) < +\infty,	\\
\EE\left[e^{\langle u_i(T)+u_0(T),X_T\rangle}\right]
&= e^{-v_i(T)-v_0(T)}\EE\left[\frac{S^{\delta_i}(T,T)}{B_T}\right]
= e^{-v_i(T)-v_0(T)}B(0,T)\EE^{\QQ^T}[S^{\delta_i}(T,T)]	\\
&= e^{-v_i(T)-v_0(T)}B(0,T)S^{\delta_i}(0,T) < +\infty.
\end{align*}
Together with Remark \ref{rem:convex_expmom}, this implies that $u_0(T)\in\mfU_T$ and $u_i(T)+u_0(T)\in\mfU_T$, for all $T\in[0,\T]$ and $i=1,\ldots,m$, thus proving that $(X,\mathbf{u},\mathbf{v})$ is an affine multi-curve model.
\end{proof}

In particular, Proposition \ref{prop:QAffineCharacterization} implies that affine multi-curve models admit explicit valuation formulas for all linear interest rate derivatives (i.e., FRA contracts, interest rate swaps, overnight swaps and basis swaps), since their prices can always be expressed in terms of OIS bond prices and multiplicative spreads, as detailed in  Appendix~\ref{sec:noopt}. 

A major issue when modeling multiple curves consists in characterizing when the multiplicative spreads are greater than one and ordered with respect to the tenor's length, as it is the case in typical market scenarios. To this effect, we can establish the following result (compare also with~\cite[Corollary~3.17]{CFG:14}).

\begin{proposition}\label{prop:X0Y}
Let $(X,\mathbf{u},\mathbf{v})$ be an affine multi-curve model.
Suppose that the process $X$ is of the form $X=(X^0,Z)$ and takes values in a state space $D=D_{X^0} \times C_Z$, with $C_Z$ a cone. Suppose that $v_i(t)\geq 0$, for all $t\in[0,\T]$ and that the functions $u_i$ are of the form $u_i=(0,w_i)$ with $w_i:[0,\T]\rightarrow C_Z^*$, where $C_Z^*$ denotes the dual cone of $C_Z$, for all $i=1,\ldots,m$. Then $S^{\delta_i}(t,T) \geq 1$ for all $0\leq t \leq T\leq \T$ and $i\in \{1, \ldots,m\}$. Moreover, if in addition
\[
v_1(t) \leq v_2(t) \leq \ldots \leq v_m(t)
\qquad\text{and}\qquad
w_1(t) \prec w_2(t) \prec \ldots \prec w_m(t),
\qquad \text{ for all }t\in[0,\T],
\]
with $\prec$ denoting the partial order on $C_Z^*$, then it holds that 
\[
S^{\delta_1}(t,T)\leq S^{\delta_2}(t,T) \leq \ldots \leq S^{\delta_m}(t,T),
\qquad \text{ for all } 0\leq t \leq T \leq \T.
\]
\end{proposition} 

\begin{proof}
The assertion is a direct consequence of Proposition \ref{prop:bonds_spreads}, stating that the process $(S^{\delta_i}(t,T))_{0\leq t\leq T}$ is a $\QQ^T$-martingale, for every $i=1,\ldots,m$ and $T\in[0,\T]$. It thus satisfies
\[
S^{\delta_i}(t,T)=\EE^{\QQ^T}[S^{\delta_i}(T,T)|\cF_t]=\EE^{\QQ^T}[e^{v_i(T) +\langle w_i(T) , Z_T \rangle}|\cF_t],
\qquad \text{ for all } t\in[0,\T].
\]
\end{proof}

Besides the modeling flexibility and tractability ensured by affine processes, this facility of generating spreads which are greater than one and ordered with respect to the tenor's length represents one of the main advantages of specifying multiplicative spreads via~\eqref{eq:log_spread}.

The family of functions $\mathbf{v}=(v_0,v_1,\ldots,v_m)$ appearing in Definition \ref{def:aff_model} can always be chosen in such a way to automatically achieve an exact fit to the initially observed term structures of OIS bond prices and multiplicative spreads. 
We denote by $B^M(0,T)$ the price of an OIS zero-coupon bond with maturity $T$ observed on the market at the initial date $t=0$ and, similarly, by $S^{M,\delta_i}(0,T)$ the multiplicative spread of tenor $\delta_i$ observed on the market at $t=0$, for $i=1,\ldots,m$ and $T\in[0,\T]$. We say that \emph{an exact fit to the initially observed term structures is achieved} if 
\[
B(0,T) = B^M(0,T)
\qquad\text{and}\qquad
S^{\delta_i}(0,T) = S^{M,\delta_i}(0,T),
\qquad \text{ for all } T\in[0,\T] \text{ and } i=1,\ldots,m.
\]
The following proposition provides a necessary and sufficient condition for an exact fit to hold. If $(X,\mathbf{u},\mathbf{v})$ is an affine multi-curve model, we denote by $B^0(t,T)$ and $S^{0,\delta_i}(t,T)$, for $0\leq t\leq T\leq \T$ and $i=1,\ldots,m$, the theoretical bond prices and spreads computed according to the model $(X,\mathbf{u},\mathbf{0})$ via \eqref{eq:generalAffineBond}-\eqref{eq:correspondence}, i.e., with all the functions $v_i$, $i=0,1,\ldots,m$, set equal to zero.

\begin{proposition}	\label{prop:fitting}
An affine multi-curve model $(X,\mathbf{u},\mathbf{v})$ achieves an exact fit to the initially observed term structures if and only if the family of functions $\mathbf{v}=(v_0,v_1,\ldots,v_m)$ satisfies
\begin{align*}
v_0(t) &=	\log B^M(0,t) - \log B^0(0,t),
&\qquad &\text{ for all }t\in[0,\T],\\
v_i(t) &=  \log S^{M,\delta_i}(0,t) - \log S^{0,\delta_i}(0,t),
&\qquad &\text{ for all }t\in[0,\T] \text{ and }i=1,\ldots,m,
\end{align*}
with $B^0(0,t)$ and $S^{0,\delta_i}(0,t)$ denoting the theoretical bond prices and spreads computed according to the affine multi-curve model $(X,\mathbf{u},\mathbf{0})$ via \eqref{eq:generalAffineBond}-\eqref{eq:correspondence}, for $t\in[0,\T]$ and $i=1,\ldots,m$.
\end{proposition}

\begin{proof}
The claim follows by noting that, for all $t\in[0,\T]$ and $i=1,\ldots,m$,
\begin{align*}
B(0,t) &= \EE\left[\frac{1}{B_t}\right] = e^{v_0(t)}\EE[e^{\langle u_0(t),X_t\rangle}] = e^{v_0(t)}B^0(0,t),	\\
S^{\delta_i}(0,t) &= \EE^{\QQ^t}[S^{\delta_i}(t,t)]
= e^{v_i(t)}\EE^{\QQ^t}[e^{\langle u_i(t),X_t\rangle}] = e^{v_i(t)}S^{0,\delta_i}(0,t),
\end{align*}
where in the last equality we used the fact that the density of the $t$-forward measure $\QQ^t$ is the same for both models $(X,\mathbf{u},\mathbf{v})$ and $(X,\mathbf{u},\mathbf{0})$. Indeed, for every $t\in[0,\T]$, it holds that
\[
\frac{1}{B_tB(0,t)} = \frac{e^{v_0(t)+\langle u_0(t),X_t\rangle}}{\EE[e^{v_0(t)+\langle u_0(t),X_t\rangle}]}
= \frac{e^{\langle u_0(t),X_t\rangle}}{\EE[e^{\langle u_0(t),X_t\rangle}]}
= \frac{1}{B^0_tB^0(0,t)},
\]
with $B^0=(B^0_t)_{0\leq t\leq \T}$ denoting the num\'eraire process corresponding to the model $(X,\mathbf{u},\mathbf{0})$.
\end{proof}

\subsection{Affine short rate multi-curve models}		\label{sec:short_rate}

We have so far introduced a general setup based on affine processes for modeling multiple curves. As outlined in the previous subsection, depending on the choice of the num\'eraire-martingale measure couple $(B,\QQ)$, different modeling approaches can be obtained from Definition \ref{def:aff_model}. We now specialize the above general framework to a classical setting where the num\'eraire is given by the OIS bank account, defined with respect to the (risk-free) OIS short rate process $r=(r_t)_{0\leq t\leq \T}$, and $\QQ$ is the associated (spot) martingale measure. More specifically, in this subsection we shall work under the following assumption.

\begin{assumption}	\label{ass:short_rate}
The num\'eraire $B=(B_t)_{0\leq t\leq \T}$ satisfies $B_t=\exp(\int_0^tr_udu)$, for all $t\in[0,\T]$, where $r=(r_t)_{0\leq t\leq \T}$ is a real-valued adapted process denoting the OIS short rate. The probability measure $\QQ$ is such that the $B$-discounted price of every basic traded asset is a $\QQ$-martingale.
\end{assumption}

Throughout this subsection, let $X=(X_t)_{0\leq t\leq \T}$ be an affine process taking values in $D_X\subseteq V_X$. We also introduce the auxiliary affine process $Y=(Y_t)_{0\leq t\leq \T}$ defined by $Y_t:=(X_t,\int_0^tX_udu)$, for $t\in[0,\T]$, and taking values in $D_Y:=D_X\times D_X$. The set $\mfU^Y$ is defined as in \eqref{eq:expmom} with respect to $Y$. We specialize Definition \ref{def:aff_model} as follows.

\begin{definition}	\label{def:short_rate_model}
Let $\lambda\in V_X$, $\ell: [0, \mathbb{T}] \to \mathbb{R}$ such that $\int_0^{\mathbb{T}} |\ell(u)| du  < +\infty $,  $\boldsymbol{\gamma}=(\gamma_1,\ldots,\gamma_m)\in V_X^m$ and $\mathbf{c}=(c_1,\ldots,c_m)$ a family of functions $c_i: [0, \mathbb{T}] \to \mathbb{R}$, $i=1, \ldots, m$. The tuple $(X,\ell,\lambda,\mathbf{c},\boldsymbol{\gamma})$ is said to be an \emph{affine short rate multi-curve model} if $(Y,\mathbf{u},\mathbf{v})$ is an affine multi-curve model in the sense of Definition \ref{def:aff_model}, where
\begin{align*}
v_0(t) &:= -\int_0^t \ell(u) du,
& \text{and }\quad
&u_0(t) := (0,-\lambda),
&\qquad&\text{ for all }t\in[0,\T],	\\
v_i(t) &:= c_i(t)
& \text{and }\quad
&u_i(t) := (\gamma_i,0),
&\qquad&\text{ for all }t\in[0,\T] \text{ and }i=1,\ldots,m.
\end{align*}
\end{definition}

If Assumption \ref{ass:short_rate} holds, so that $B=\exp(\int_0^{\cdot}r_udu)$, assuming that a tuple $(X,\ell,\lambda,\mathbf{c},\boldsymbol{\gamma})$ is an affine short rate multi-curve model in the sense of Definition \ref{def:short_rate_model} is equivalent to letting the OIS short rate $r_t$ and the spot multiplicative spreads $S^{\delta_i}(t,t)$ be specified as follows:
\begin{align}
r_t &= \ell(t) + \langle\lambda,X_t\rangle,
&\text{ for all }&t\in[0,\T],	\label{eq:short_rate_1}\\
\log S^{\delta_i}(t,t) &= c_i(t) + \langle\gamma_i,X_t\rangle,
 &\text{ for all }&t\in[0,\T] \text{ and } i=1,\ldots,m. \label{eq:short_rate_2}
\end{align}
Note that Definition \ref{def:short_rate_model} implicitly requires (via Definition \ref{def:aff_model}) that the parameters $\lambda$ and $\boldsymbol{\gamma}=(\gamma_1,\ldots,\gamma_m)$ satisfy $(\gamma_i,-\lambda)\in\mfU^Y$, for all $i=0,1,\ldots,m$, with $\gamma_0:=0$. 

\begin{remark}
If the function $\ell$ is chosen to be constant and $\mathbf{c}$ and $\boldsymbol{\gamma}$ set to $0$, Definition \ref{def:short_rate_model} reduces to a classical (i.e., single-curve) time-homogeneous affine short rate model. By additionally allowing $\ell$ and also $\mathbf{c}$ to be deterministic functions we introduce a time-inhomogeneity which in turn enables a perfect fit to the initial term structures (see Proposition~\ref{prop:deterministicShiftCharacterization} below). Definition \ref{def:short_rate_model} thus already incorporates  \emph{deterministic shift extended models} first introduced in~\cite{bm01} in the single-curve setting and recently taken up in~\cite{GM:14} in the multi-curve setting (cf. Section~\ref{sec:rel_short_rate} for the relation between~\cite{GM:14} and the present setting).
\end{remark}

\begin{remark}
In Definition \ref{def:short_rate_model}, the OIS bank account is specified via a short rate modeled as in \eqref{eq:short_rate_1}. By analogy,  multiplicative spreads can also be specified via an instantaneous \emph{short spread rate}, similarly as in \cite[Remark~3.19]{CFG:14}. This can be achieved by letting the affine process $X$ be of the form $X=(X^0,Z^1,\ldots,Z^m)$, where $X^0$ is an affine process on some state space $D_{X^0} $ and $Z^i:=\int_0^{\cdot} q^i(X^0_s) ds$, with $q^i : D_{X^0} \to \RR$ being an affine function, for all $i=1,\ldots,m$. The coefficients $\mathbf{c}=(c_1,\ldots,c_m)$ and $\boldsymbol{\gamma}=(\gamma_1,\ldots,\gamma_m)$ are chosen as $c_i:=0$ and $\gamma_i:=(0,e_i)$, for all $i=1,\ldots,m$, with $\{e_1,\ldots,e_m\}$ being the canonical basis of $\RR^m$. 
The spot multiplicative spreads admit then the representation
\[
S^{\delta_i}(t,t)=e^{Z_t^i}=e^{\int_0^t  q^i(X^0_s) ds},
\qquad \text{ for all }t\in[0,\T] \text{ and }i=1,\ldots,m,
\]
where the process $(q^i(X_t^0))_{0\leq t\leq \T}$ can be interpreted as a short spread rate, for $i=1,\ldots,m$.
\end{remark}

Similarly as in the case of general affine multi-curve models, affine short rate multi-curve models admit an equivalent characterization in terms of the  structure of OIS bond prices and spreads. More specifically, we can establish the following result, which specializes Proposition \ref{prop:QAffineCharacterization}.

\begin{proposition}
Let $X$ be an affine process and suppose that Assumption \ref{ass:short_rate} holds. Then the following hold:
\begin{enumerate}
\item 
if $(X,\ell,\lambda,\mathbf{c},\boldsymbol{\gamma})$ is an affine short rate multi-curve model, then $X$ generates exponentially affine OIS bond prices and spreads, in the sense that
\begin{equation}	\label{eq:time_homogeneous_bond_spreads}	\begin{aligned}
B(t,T) &= \exp\left(\widetilde{\cA}^0(t,T) + \langle\widetilde{\cB}^0(T-t),X_t\rangle\right),
&\quad&\text{ for all }0\leq t\leq T\leq \T,	\\
S^{\delta_i}(t,T) &= \exp\left((\widetilde{\cA}^i(t,T) + \langle\widetilde{\cB}^i(T-t),X_t\rangle\right),
&\quad&\text{ for all }0\leq t\leq T\leq \T \text{ and }i=1,\ldots,m,
\end{aligned}	\end{equation}
where the functions $\widetilde{\cA}^i$, $\widetilde{\cB}^i$, $i=0,1,\ldots,m$, are given by
\begin{align*}
\widetilde{\cA}^0(t,T) &= -\int_t^T\ell(u)du + \tilde{\phi}(T-t,0,-\lambda),	
&\quad &\widetilde{\cB}^0(T-t) = \tilde{\psi}(T-t,0,-\lambda),	\\
\widetilde{\cA}^i(t,T) &= c_i(T) + \tilde{\phi}(T-t,\gamma_i,-\lambda)-\tilde{\phi}(T-t,0,-\lambda),	
&\quad &\widetilde{\cB}^i(T-t) = \tilde{\psi}(T-t,\gamma_i,-\lambda)-\tilde{\psi}(T-t,0,-\lambda),
\end{align*}
for all $0\leq t\leq T\leq\T$ and $i=1,\ldots,m$, with $\tilde{\phi}$ and $\tilde{\psi}$ denoting the characteristic exponents of the affine process $Y=(X,\int_0^{\cdot}X_udu)$ as in Lemma \ref{mixedstate};
\item
conversely, if OIS bond prices and spreads admit the representation \eqref{eq:time_homogeneous_bond_spreads} for some functions $\widetilde{\cA}^i$, $\widetilde{\cB}^i$, $i=0,1,\ldots,m$, with $T \mapsto \widetilde{\cA}^0(t,T)$ and $T\mapsto\widetilde{\cB}^0(T)$   differentiable, for all $t\in[0,\T]$, then $(X,\ell,\lambda,\mathbf{c},\boldsymbol{\gamma})$ is an affine short rate model, where, for all $i=1,\ldots,m$,
\[
\ell(t) = -\frac{\partial}{\partial T}\bigl(\widetilde{\cA}^0\bigr)(t,T)|_{T=t},
\qquad \lambda = -\bigl(\widetilde{\cB}^0\bigr)'(0),
\qquad c_i(t) = \widetilde{\cA}^i(t,t)
\qquad\text{ and }
\qquad \gamma_i = \widetilde{\cB}^i(0).
\]
\end{enumerate}
\end{proposition}

\begin{proof}
Noting that the affine process $Y=(X,\int_0^{\cdot}X_udu)$ satisfies the assumptions of Lemma \ref{mixedstate}, part (i) easily follows from Definition \ref{def:short_rate_model} together with Proposition \ref{prop:QAffineCharacterization}.
Conversely, if OIS bond prices and spreads admit the representation \eqref{eq:time_homogeneous_bond_spreads} with respect to some differentiable functions $\widetilde{\cA}^i$ and $\widetilde{\cB}^i$, $i=0,1,\ldots,m$, then the instantaneous OIS forward rate is given by
\[
f_t(T) := -\frac{\partial}{\partial T} \log B(t,T) = -\frac{\partial}{\partial T}(\widetilde{\cA}^0)(t,T) - \langle(\widetilde{\cB}^0)'(T-t),X_t\rangle,
\qquad \text{ for all }0\leq t\leq T\leq \T,
\]
so that the OIS short rate $r_t$ satisfies
\[
r_t = f_t(t) =  -\frac{\partial}{\partial T}(\widetilde{\cA}^0)(t,T)|_{T=t} - \langle(\widetilde{\cB}^0)'(0),X_t\rangle
=: \ell(t) + \langle\lambda,X_t\rangle,
\qquad \text{ for all }t\in[0,\T].
\]
Representations \eqref{eq:short_rate_1}-\eqref{eq:short_rate_2} then follow by letting $c_i(t):=\widetilde{\cA}^i(t,t)$ and $\gamma_i:=\widetilde{\cB}^i(0)$, for all $i=1,\ldots,m$. Finally, the condition $(\gamma_i,-\lambda)\in\mfU^Y$, for all $i=0,1,\ldots,m$, with $\gamma_0:=0$, follows by the same arguments used in the last part of the proof of Proposition \ref{prop:QAffineCharacterization}.
\end{proof}

In view of Proposition \ref{prop:X0Y}, multiplicative spreads greater than one and ordered with respect to the tenor's length can be easily generated by affine short rate multi-curve models.

\begin{remark}[On the possibility of negative OIS short rates and spreads greater than one]
In the current market environment, negative rates have been observed to coexist with spreads which are strictly greater than one and ordered with respect to the tenor's length. This market scenario can be easily captured by letting the affine process $X$ be of the form $X=(X^0,Z)$ and take value in a state space of the form $\RR^n\times C_Z$, for some $n\in\mathbb{N}$ and where $C_Z$ is a cone. If $\ell$ takes values in $\RR$, $\lambda\in\RR^n$ and the assumptions of Proposition \ref{prop:X0Y} are satisfied, then spreads will be greater than one and ordered, while the OIS short rate will not be restricted to positive values only. 
\end{remark}

Similarly to Proposition \ref{prop:fitting}, we can obtain a necessary and sufficient condition to achieve an exact fit to the initially observed term structures. Let us denote by $B^0(t,T)$ and $S^{0,\delta_i}(t,T)$ the bond prices and spreads computed according to the model $(X,0,\lambda,\mathbf{0},\boldsymbol{\gamma})$. Moreover, we denote by $T\mapsto B^M(0,T)$ and $\{T\mapsto S^{M,\delta_i}(0,T); i=1,\ldots,m\}$ the initially observed term structures of OIS bond prices and spreads. We also define the corresponding OIS forward rates by
\[
f^M_t(T) := -\frac{\partial}{\partial T}\log B^M(t,T)
\qquad\text{and}\qquad
f^0_t(T) := -\frac{\partial}{\partial T}\log B^0(t,T),
\qquad \text{ for all }0\leq t\leq T\leq \T.
\]
The following result is a direct consequence of the above definitions together with Proposition \ref{prop:fitting}. 

\begin{proposition}	\label{prop:deterministicShiftCharacterization}
Suppose that Assumption \ref{ass:short_rate} holds. Then an affine short rate multi-curve model $(X,\ell,\lambda,\mathbf{c},\boldsymbol{\gamma})$ achieves an exact fit to the initially observed term structures if and only if 
\begin{align*}
\ell(t) &= f^M_0(t) - f^0_0(t),
&\qquad &\text{ for all }t\in[0,\T],\\
c_i(t) &= \log S^{M,\delta_i}(0,t) - \log S^{0,\delta_i}(0,t),
&\qquad &\text{ for all }t\in[0,\T] \text{ and }i=1,\ldots,m.
\end{align*}
\end{proposition}

\begin{remark}
In view of Proposition~\ref{prop:deterministicShiftCharacterization}, 
assuming that the spreads $S^{0,\delta_i}(t,T)$ generated by the model $(X,0,\lambda,\mathbf{0},\boldsymbol{\gamma})$ are greater than one and ordered, it then easily follows that
\begin{enumerate}
\item[(i)]
for any $i\in\{1,\ldots,m\}$, if $S^{M,\delta_i}(0,t) \geq S^{0,\delta_i}(0,t)$, for all $t\in [0, \T]$, then $S^{\delta_i}(t,T)\geq 1$, for all $0\leq t\leq T\leq \T$;
\item[(ii)]
for any $i,j\in\{1,\ldots,m\}$ with $i<j$, if 
\[
\log S^{M,\delta_i}(0,t)- \log S^{M,\delta_j}(0,t) \leq \log S^{0,\delta_i}(0,t)- \log S^{0,\delta_j}(0,t),
\qquad \text{ for all } t\in[0,\T],
\]
then $S^{\delta_i}(t,T)\leq S^{\delta_j}(t,T)$, for all $0\leq t\leq T\leq \T$.
\end{enumerate}
\end{remark}

\subsection{Relations with existing multi-curve models based on affine processes}
\label{sec:relations}

In this section, we discuss how our setup relates to several multiple curve models based on affine processes that have been recently proposed in the literature. We show that all the existing approaches we are aware of can be recovered as special cases of our  framework.

\subsubsection{Short rate models}	\label{sec:rel_short_rate}

As mentioned in the introduction, short rate models based on affine processes have been proposed for modeling multiple curves, see \cite{fitr12,GM:14,GMR:15,ken10,kitawo09,MR14}. A general framework   encompassing all the models introduced in these papers has been recently provided in \cite[Chapter 2]{GR15}. In that work, the authors assume that the num\'eraire is given as usual by the OIS bank account, modeled via a short rate $r=(r_t)_{0\leq t\leq \T}$, with a corresponding martingale measure $\QQ$. Considering for simplicity of presentation a single tenor $\delta>0$, spot Libor rates $L_t(t,t+\delta)$ are then specified by
\[
L_t(t,t+\delta) = \frac{1}{\delta}\left(\frac{1}{B^{\delta}(t,t+\delta)}-1\right),
\qquad \text{ for all }0\leq t\leq \T,
\]
where $B^{\delta}(t,t+\delta)$ denotes the price at date $t$ of an artificial risky bond with maturity $t+\delta$. Note that risky bonds are only introduced as a modeling tool, their riskiness being referred to the interbank market, and do not represent real traded assets. 
By analogy to the short rate approach in a classical single-curve setting, OIS bond prices and risky bond prices are assumed to be given by
\[
B(t,T) = \mathbb{E}\left[e^{-\int_t^Tr_udu}\bigr|\cF_t\right]
\quad\text{ and }\quad
B^{\delta}(t,T) = \mathbb{E}\left[e^{-\int_t^T(r_u+s_u)du}\bigr|\cF_t\right],
\]
for all $0\leq t\leq T\leq \T$, where $s=(s_t)_{0\leq t\leq \T}$ denotes an instantaneous spread. The OIS short rate $r_t$ and the spread $s_t$ are modeled as affine functions of a common affine process $X=(X_t)_{0\leq t\leq \T}$ taking values on some state space $D\subseteq V:=\RR^n_+\times\RR^{d-n}$ (\cite{GM:14} also consider the case of a Wishart process $X$ and deterministic shift extensions, meaning that the constant part of the affine function is actually a deterministic function in time). 
This framework can be regarded as a special case of our affine short rate setup introduced in Section \ref{sec:short_rate}. Indeed, it suffices to observe that in the present setting the spot multiplicative spreads can be represented as
\[
\log S^{\delta}(t,t) 
= \log \frac{B(t,t+\delta)}{B^{\delta}(t,t+\delta)}
= \log \frac{\mathbb{E}\left[e^{-\int_t^{t+\delta}r_udu}\bigr|\cF_t\right]}{\mathbb{E}\left[e^{-\int_t^{t+\delta}(r_u+s_u)du}\bigr|\cF_t\right]}
= c_{\delta}(t) + \langle \gamma_{\delta},X_t\rangle,
\]
where $c_{\delta}:[0,\T]\rightarrow\RR$ and $\gamma_{\delta} \in V$ are determined in terms of the characteristic exponents of $X$ and the specification of the affine functions for $r_t$ and $s_t$. This shows that representations \eqref{eq:short_rate_1}-\eqref{eq:short_rate_2} hold with respect to the underlying affine process $X$. 
Similar considerations allow to recover the Gaussian exponentially quadratic model recently proposed in~\cite{GMR:15}. 
Since our spot multiplicative spread is an observable quantity, unlike the instantaneous spread $s_t$, it seems to us more natural to define affine short rate multi-curve models directly with respect to spot multiplicative spreads as in Definition \ref{def:short_rate_model}.

\subsubsection{Affine Libor models}	\label{sec:rel_affineLibor}

The general affine multi-curve setup of Section \ref{sec:aff_model} allows to recover a generalization of the model recently introduced in \cite{GPSS14}, which extends to a multi-curve setting the affine Libor model originally proposed in \cite{KRPT13}.
In this setting, the num\'eraire - martingale measure couple $(B, \QQ)$ is  given by the OIS zero-coupon bond with maturity $\T$ and the corresponding $\T$-forward measure, i.e., $(B, \QQ)= (B(\cdot, \T), \QQ^{\T})$. We
explicitly indicate this measure in the expectations below. 
As in Section \ref{sec:aff_model}, consider an affine process  $X=(X_t)_{0\leq t \leq \T}$ (under $\QQ^{\T}$) taking values in $\RR^d_+$ and denote by $M^u=(M^u_t)_{0\leq t\leq \T}$ the $\QQ^{\T}$-martingale defined by $M^u_t:=\EE^{\QQ^{\T}}\left[\exp(\langle u,X_{\T}\rangle)|\cF_t\right]$, for $t\in[0,\T]$.

For each tenor $\delta>0$ traded in the market,  \cite{GPSS14} consider a finite collection of ordered maturities $\mathcal{T}^{\delta}:=\{0=T^{\delta}_0,T^{\delta}_1,\ldots,T^{\delta}_{N^{\delta}}=\T\}$. 
For any tenor $\delta>0$ and $k\in\{0,\ldots,N^{\delta}-1\}$, the multiplicative spread $S^{\delta}(t,T^{\delta}_k) $ is given by
\begin{equation}	\label{GPSS-spread}
S^{\delta}(t,T^{\delta}_k) 
= \frac{1+\delta L_t(T^{\delta}_k,T^{\delta}_{k+1})}{1+\delta L^{{\rm OIS}}_t(T^{\delta}_k,T^{\delta}_{k+1})}
= \frac{M^{v^{\delta}_k}_t}{M^{u^{\delta}_k}_t},
\qquad \text{ for all }t\in[0,T^{\delta}_k],
\end{equation}
where $\{u^{\delta}_1,\ldots,u^{\delta}_{N^{\delta}}\}$ is a sequence in $\mfU$ such that $u^{\delta}_k\geq u^{\delta}_{k+1}$, for all $k$, in order to ensure non-negative OIS forward rates, and $\{v^{\delta}_1,\ldots,v^{\delta}_{N^{\delta}}\}$ is a sequence in $\mfU$ such that $v^{\delta}_k\geq u^{\delta}_k$, for all $k$, in order to ensure that spreads are greater than one. 

By its own nature, the model of \cite{GPSS14} is formulated with respect to a finite collection of maturities. However, a continuous extension can be obtained by considering a function $u:[0,\T]\rightarrow\mfU$ such that $u(T^{\delta}_k)=u^{\delta}_k$ and a function $v:[0,\T]\times \mathbb{R}_+ \rightarrow\mfU$ such that $v(T^{\delta}_k,\delta)=v^{\delta}_k$, for all $k$ and $\delta>0$. 
The idea of a continuous extension first appeared (in a single-curve setting) in the unpublished note~\cite{KR_note} and, after the appearance of the present paper, has been further developed in the multi-curve setting in~\cite{PW:16}.
In analogy to \cite[Section 6]{KRPT13}, let OIS bond prices be given as follows:
\begin{equation}	\label{GPSS-bonds}
B(t,T)/B(t,\T) = M^{u(T)}_t,
\qquad \text{ for all } 0\leq t\leq T\leq \T,
\end{equation}
with $M^{u(T)}_t=\EE^{\QQ^{\T}}\left[\exp\left(\langle u(T),X_{\T}\rangle\right)|\cF_t\right]$. 
In line with equation \eqref{GPSS-spread}, we then model $S^{\delta}(t,T)$ by
\begin{equation}	\label{GPSS-S}
S^{\delta}(t,T) = M^{v(T,\delta)}_t/M^{u(T)}_t,
\qquad \text{ for all }0\leq t\leq T\leq \T \text{ and }\delta>0,
\end{equation}
with $M^{v(T,\delta)}_t=\EE^{\QQ^{\T}}\left[\exp\left(\langle v(T,\delta),X_{\T}\rangle\right)|\cF_t\right]$.
Observe that, in order to ensure that $S^{\delta}(t,T)\geq1$, for all $0\leq t\leq T\leq \T$ and $\delta>0$, it suffices to require that $v(T,\delta)\geq u(T)$ for all $T\in[0,\T]$ and $\delta>0$. Similarly, order relations among multiplicative spreads $S^{\delta}(t,T)$ associated to different tenors can be obtained by imposing suitable requirements on the function $\delta\mapsto v(T,\delta)$. 

Recalling that $B(t,t)=1$, for all $t\in[0,\T]$, representation \eqref{GPSS-bonds} together with the affine property of $X$ implies that the num\'eraire process satisfies
\[
\log B_t := \log B(t,\T) = - \log M^{u(t)}_t = -\phi\bigl(\T-t,u(t)\bigr) - \bigl\langle\psi\bigl(\T-t,u(t)\bigr),X_t\bigr\rangle,
\qquad \text{ for all }t\in[0,\T],
\]
where $\phi$ and $\psi$ denote the characteristic exponents of $X$. 
Similarly, representation \eqref{GPSS-S} implies that the spot multiplicative spreads satisfy, for all $t\in[0,\T]$ and $\delta>0$,
\[
\log S^{\delta}(t,t) = \log\frac{M^{v(t,\delta)}_t}{M_t^{u(t)}}
= \phi\bigl(\T-t,v(t,\delta)\bigr)-\phi\bigl(\T-t,u(t)\bigr) + \bigl\langle\psi\bigl(\T-t,v(t,\delta)\bigr)-\psi\bigl(\T-t,u(t)\bigr),X_t\bigr\rangle.
\]
The model \eqref{GPSS-bonds}-\eqref{GPSS-S} thus represents a special case of our general affine multi-curve framework. Indeed, representations \eqref{eq:log_numeraire}-\eqref{eq:log_spread} follow by letting, for all $t\in[0,\T]$ and $i=1,\ldots,m$,
\begin{equation}\label{eq:comp_GPSS}
\begin{aligned}
v_0(t)&=\phi\bigl(\T-t,u(t)\bigr),
&\qquad &u_0(t)=\psi\bigl(\T-t,u(t)\bigr),	\\
v_i(t) &= \phi\bigl(\T-t,v(t,\delta_i)\bigr)-\phi\bigl(\T-t,u(t)\bigr),
&\qquad &u_i(t) = \psi\bigl(\T-t,v(t,\delta_i)\bigr)-\psi\bigl(\T-t,u(t)\bigr).
\end{aligned}
\end{equation}
Moreover, in view of \cite[Lemma 4.3]{krm12}, it holds that $u_0(T)\in\mfU_T$ and $u_i(T)+u_0(T)\in\mfU_T$, for all $T\in[0,\T]$ and $i=1,\ldots,m$, thus showing that this version of the affine Libor model with multiple curves can be recovered as a special case of Definition \ref{def:aff_model}.

In our view, our general framework presents significant advantages with respect to the model of \cite{GPSS14}, not only due to the flexibility in the choice of the num\'eraire - martingale measure couple and of affine processes on general convex state spaces. 
Indeed, even under the assumption that the couple $(B,\QQ)$ is specified in terms of the OIS bond with maturity $\T$ together with the corresponding $\T$-forward measure as above, our Definition \ref{def:aff_model} allows for greater generality. This can be seen by noting that, given a family of functions $({\mathbf u},\mathbf{v})$ as in Definition \ref{def:aff_model}, it is not always possible to find a couple of functions $(u,v):[0,\T]\times[0,\T]\times\RR_+\rightarrow\mfU\times\mfU$ such that \eqref{eq:comp_GPSS} holds.
Moreover, as shown in Proposition \ref{prop:fitting}, our specification \eqref{eq:log_numeraire}-\eqref{eq:log_spread} always allows for an automatic fit to the initially observed term structures. In contrast, in \cite{GPSS14} this is only possible under some conditions, notably that the initial OIS term structure is decreasing with respect to maturity together with the finiteness of suitable exponential moments of the driving process. Furthermore, while in our specification the initial term structures univocally determine the family of functions $\mathbf{v}$, in~\cite{GPSS14} this is not the case, so that all model parameters have to be jointly determined by fitting the initial term structures and by calibration to market data.

\subsubsection{Affine multiple yield curve models in the sense of \cite{CFG:14}}

Multi-curve models based on affine processes have been briefly mentioned in \cite[Section 5.3]{CFG:14} as a simple example of {\em HJM-type multiple yield curve models}. The affine specification proposed in \cite{CFG:14} represents a special case of the general setup developed in the present paper. Indeed, an affine multi-curve model in the sense of Definition~\ref{def:aff_model} belongs to the family of risk neutral HJM-type multiple yield curve models considered in \cite{CFG:14} only if $\log B_t$ is absolutely continuous.
More specifically, the affine specification considered in \cite{CFG:14} (see \cite[Definition 5.3]{CFG:14}) represents a special case of our~Definition \ref{def:short_rate_model} when $\ell$ and $\mathbf{c}$ are chosen to be constant.
For clarity of notation, let us denote all the ingredients appearing in~\cite[Definition 5.3]{CFG:14} with Fractur letters.
Then~\cite[Definition 5.3]{CFG:14} can be embedded into Definition~\ref{def:short_rate_model} by setting $X:=(\mathfrak{X}, \mathfrak{Y})$, $c_i=0$ and $\gamma_i=(0, \mathfrak{u}_i)$, for all $i=1, \ldots, m$. 
Conversely, in the special case where $\ell$ and $\mathbf{c}$ are constant, Definition~\ref{def:short_rate_model} can be embedded into~\cite[Definition 5.3]{CFG:14} by letting $\mathfrak{X}=X$, $\mathfrak{Y}=(1,X)$, $\mathfrak{u}^1_i=c_i$ and $\mathfrak{u}^{j+1}_i=\gamma_i^j$, for all $j=1,\ldots,\text{dim}(V_X)$ and $i=1,\ldots,m$.

\section{General pricing formulae for caplets and swaptions} 	\label{sec:aff_pricing}

We now show that affine multi-curve models, in the sense of Definition~\ref{def:aff_model}, lead to tractable general valuation formulae for caplets and swaptions. We compute \emph{clean prices} and follow the pricing approach outlined in~\cite[Appendix A]{CFG:14}, in particular assuming that the collateral account is given by the num\'eraire asset.
For simplicity of presentation, we shall consider 
a fixed maturity $T>0$ and assume that $(X,\mathbf{u},\mathbf{v})$ is an affine multi-curve model in the sense of Definition~\ref{def:aff_model}.

\subsection{Caplets}\label{sec:capletsAffine}

In the present affine setting, caplets can be easily priced by means of Fourier techniques. As a preliminary, let $i \in \{1, \ldots,m\}$ and define the stochastic process $(\mathcal{Y}_t)_{t\in[0,T+\delta_i]}$ by
\begin{equation}	\label{eq:proc_Y_gen}
\begin{aligned}
\mathcal{Y}_t&:=\log\left(\frac{S^{\delta_i}(t,t)}{B(t,T+\delta_i)}\right)\\
&=v_i(t)+v_0(t)-v_0(T+\delta_i)-\phi(T+\delta_i-t, u_0(T+\delta_i))	\\
&\quad+ \langle u_i(t)+ u_0(t)-\psi(T+\delta_i-t, u_0(T+\delta_i)), X_t\rangle,
\end{aligned}
\end{equation} 
where the second equality follows from Proposition~\ref{prop:QAffineCharacterization}.
We denote by $\mathcal{A}_{\mathcal{Y}_T}$ the set 
\[
\mathcal{A}_{\mathcal{Y}_T}:=\left\{\nu \in \RR: \,  \mathbb{E}\left[\frac{B(T,T+\delta_i)}{B_T}e^{\nu\mathcal{Y}_T}\right]<+\infty\right\}^{\circ},
\]
and introduce the strip of complex numbers $\Lambda_{\mathcal{Y}_T}=\left\{\zeta \in \CC: \ -\Im(\zeta)\in\mathcal{A}_{\mathcal{Y}_T} \right\}$.
For $\zeta\in\Lambda_{\mathcal{Y}_T}$, we can compute the following expectation, which we call \emph{modified moment generating function} of $\mathcal{Y}_T$:
\begin{align}
\begin{split}
\varphi_{\mathcal{Y}_{T}}(\zeta)&:=\Excond{}{\frac{B(T,T+\delta_i)}{B_T}e^{\im \zeta \mathcal{Y}_T}}{\cF_t}\\
&=\exp((1- \im \zeta)(v_0(T+\delta_i)+\phi(\delta_i, u_0(T+\delta_i)))+ \im \zeta(v_i(T)+v_0(T)))\\
&\quad \times\Excond{}{e^{\langle (1- \im \zeta) \psi(\delta_i, u_0(T+\delta_i))+  \im \zeta (u_i(T)+ u_0(T)), X_T\rangle}}{\cF_t}\\
&=\exp((1- \im \zeta)(v_0(T+\delta_i)+\phi(\delta_i, u_0(T+\delta_i)))+ \im \zeta(v_i(T)+v_0(T)))\\
&\quad \times \exp(\phi(T-t, (1- \im \zeta) \psi(\delta_i, u_0(T+\delta_i))+  \im \zeta (u_i(T)+ u_0(T))))\\
&\quad \times \exp(\langle\psi(T-t, (1- \im \zeta) \psi(\delta_i, u_0(T+\delta_i))+  \im \zeta (u_i(T)+ u_0(T))),X_t\rangle).
\label{cfchi}
\end{split}
\end{align}

\begin{remark}
In the case of an affine short rate multi-curve model as of Definition~\ref{def:short_rate_model}, expression \eqref{eq:proc_Y_gen} becomes
\[
\mathcal{Y}_t=c_i(t)+\int_t^{T+\delta_i} \ell(u) du- \tilde{\phi}(T+\delta_i-t, 0,-\lambda)+\langle \gamma_i-\tilde{\psi}(T+\delta_i-t, 0, -\lambda), X_t\rangle.
\]
Similarly, expression \eqref{cfchi} becomes
\begin{align*}
\begin{split}
\varphi_{\mathcal{Y}_{T}}(\zeta)
&=\exp\left((1- \im \zeta)\left(-\int_0^{T+\delta} \ell(u) du+\tilde{\phi}(\delta_i,0,-\lambda)\right)+ \im \zeta(c_i(T)-\int_{0}^{T}\ell(u)du)\right)\\
&\quad \times \exp (\tilde{\phi}(T-t,  (1- \im \zeta) \tilde{\psi}(\delta_i,0,-\lambda)+  \im \zeta \gamma_i, -\im \zeta\lambda))\\
&\quad \times \exp\left(\langle \tilde{\psi}(T-t,  (1- \im \zeta) \tilde{\psi}(\delta_i,0,-\lambda)+  \im \zeta \gamma_i,-\im \zeta\lambda), X_t\rangle-\langle  \lambda, \int_0^tX_s ds\rangle\right).
\end{split}
\end{align*}
\end{remark}

The sets $\mathcal{A}_{\mathcal{Y}_T}$ and $\Lambda_{\mathcal{Y}_T}$ depend on the specific choice of the driving process $X$. By relying on~\cite[Theorem 5.1]{lee2004}, we now provide a general caplet pricing formula, which is valid for any choice of the underlying affine process $X$ and for different choices of the contour of integration. In particular, the next result highlights the tractability of our framework: caplets can be priced by means of univariate Fourier integrals. In turn, this implies that a calibration may be obtained with a reasonable amount of computational effort (which may be further reduced by means of an application of an FFT algorithm).

\begin{proposition}\label{prop:capletBasicAffine}
Let $\zeta\in\CC$, $\epsilon\in\RR$, $\bar{K}:=1+\delta K$ and assume that $1+\epsilon \in\mathcal{A}_{\mathcal{Y}_T}$. Then the price at date $t$ of a caplet with notional $N$, reset date $T$ and payoff $N\delta_i(L_T(T,T+\delta_i)-K)^+$ at the settlement date $T+\delta_i$ is given by
\begin{equation}
\Pi^{CPLT}(t; T,T+\delta_i,K,N)
=NB_t\left(R\left(\mathcal{Y},\bar{K},\epsilon\right)+\frac{1}{\pi }\int_{0-\im\epsilon}^{\infty-\im\epsilon}\Re\left(e^{-\im\zeta \log(\bar{K})}\frac{\varphi_{\mathcal{Y}_{T}}(\zeta-\im)}{-\zeta(\zeta-\im)}\right)d\zeta\right),  
\label{capletprice}
\end{equation}
where $\varphi_{\mathcal{Y}_T}$ is given in \eqref{cfchi} and $R\left(\mathcal{Y},\bar{K},\epsilon\right)$ is given by
\[
R\left(\mathcal{Y},\bar{K},\epsilon\right) = 
\begin{cases}
 \varphi_{\mathcal{Y}_{T}}(-\im)-\bar{K}\varphi_{\mathcal{Y}_{T}}(0), & \mbox{if } \epsilon<-1, \\ 
 \varphi_{\mathcal{Y}_{T}}(-\im)-\frac{\bar{K}}{2}\varphi_{\mathcal{Y}_{T}}(0), & \mbox{if } \epsilon=-1, \\
   \varphi_{\mathcal{Y}_{T}}(-\im) & \mbox{if } -1<\epsilon<0,\\
    \frac{1}{2}\varphi_{\mathcal{Y}_{T}}(-\im)&\mbox{if } \epsilon = 0,\\
    0&\mbox{if } \epsilon >0.
 \end{cases}
\]
\end{proposition}
\begin{proof}
As shown in Appendix~\ref{sec:option}, the price of a caplet can be expressed as
\begin{align*}
\Pi^{CPLT}(t; T,T+\delta_i,K,N)&=N\, \Excond{}{\frac{B_t}{B_T}\left(S^{\delta_i}(T,T)-(1+\delta_i K)B(T,T+\delta_i)\right)^+}{\cF_t}\\
&=N\,\Excond{}{\frac{B_t}{B_T}B(T,T+\delta_i)\left(S^{\delta_i}(T,T)B(T,T+\delta_i)^{-1}-(1+\delta_i K)\right)^+}{\cF_t}\\
&=N\,\Excond{}{\frac{B_t}{B_T}B(T,T+\delta_i)\left(e^{\mathcal{Y}_T}-(1+\delta_i K)\right)^+}{\cF_t}
\end{align*}
Since the modified moment generating function of $\mathcal{Y}_T$ can be explicitly computed as in \eqref{cfchi}, the pricing of a caplet is thus reduced to the pricing of a call option written on an asset whose characteristic function is explicitly known. At this point, a direct application of \cite[Theorem 5.1]{lee2004} yields the result. Note that, in the terminology of \cite[Theorem 5.1]{lee2004}, the present case corresponds to $G=G_1$ and $b_1=1$, which is an element of $\mathcal{A}_{\mathcal{Y}_T}$ by Definition~\ref{def:aff_model}.
\end{proof}

\subsection{Swaptions} \label{swaptions}

In the present general setting, swaptions do not admit a closed-form pricing formula. Indeed, on the one side, we consider general multi-factor models, so that the ``Jamshidian trick" (see \cite{jam89}) is not applicable; on the other side, in a multiple curve setting, a payer (resp. receiver) swaption cannot be represented as a put (resp. call) option on a coupon bond. 
In this subsection, by relying on Fourier methods and along the lines of \cite{cfg15,cfgg13}, we provide a general analytical approximation which exploits the affine property of our framework. We work in the general setting of Section~\ref{sec:aff_model}, noting that all formulas admit suitable simplifications in the case of affine short-rate multi-curve models.

We consider a European payer swaption with maturity $T$, written on a (payer) interest rate swap starting at $T_0=T$, with payment dates $T_1,\ldots, T_\cN$, with $T_{j+1}-T_j=\delta_i$ for $j=1,\ldots,\cN-1$ and $i \in \{1, \ldots,m\}$, with notional $N$. As shown in Appendix~\ref{sec:option}, the value of such a claim at date $t$ is given by
\[
\Pi^{SWPTN}(t; T_{1},T_{\cN},K,N)\\
=N\mathbb{E}^{}\left[\frac{B_t}{B_T}\left(\sum_{j=1}^{\cN}B(T,T_{j-1})S^{\delta_i}(T,T_{j-1})-(1+\delta_i K)B(T,T_j)\right)^+\biggr|\cF_t\right].
\]
The general idea underlying our approximation consists in approximating the exercise region by an event defined in terms of an affine function of $X_T$. More specifically, we have that
\begin{align*}
\Pi^{SWPTN}(t; T_{1},T_{\cN},K,N)&\geq N\mathbb{E}^{}\left[\frac{B_t}{B_T}\left(\sum_{j=1}^{\cN}B(T,T_{j-1})S^{\delta_i}(T,T_{j-1})-(1+\delta_i K)B(T,T_j)\right)^+\ind_\cG\biggr|\cF_t\right]\\
&\geq N\mathbb{E}^{}\left[\frac{B_t}{B_T}\left(\sum_{j=1}^{\cN}B(T,T_{j-1})S^{\delta_i}(T,T_{j-1})-(1+\delta_i K)B(T,T_j)\right)\ind_\cG\biggr|\cF_t\right]\\
&=:\widetilde{\Pi}^{SWPTN}(\alpha, \beta),
\end{align*}
where 
$\cG:=\left\{\left.\omega\in\Omega\right| \ \left\langle\beta,X_T\right\rangle > \alpha\right\}$ and  $\beta \in V_X,\alpha\in\RR$.
Let us simplify the notation by introducing
 \begin{align*}
 & w_{i,j}:=\begin{cases}
    1, & j=1,\ldots, \cN,\\
    -(1+\delta_i K), & j=\cN+1,\ldots, 2\cN,
  \end{cases}
 & \qquad u_{i,j}(T):=\begin{cases}
    u_{i}(T), & j=1,\ldots, \cN,\\
    0, & j=\cN+1,\ldots, 2\cN,
  \end{cases}\\
   &v_{i,j}(T):=\begin{cases}
    v_i(T), & j=1,\ldots, \cN,\\
    0, & j=\cN+1,\ldots, 2\cN,
  \end{cases}
 & \qquad
   l:=\begin{cases}
    j-1, & j=1,\ldots, \cN,\\
    j, & j=\cN+1,\ldots, 2\cN,
  \end{cases}
\end{align*}
so that, in view of Proposition \ref{prop:QAffineCharacterization}, $\widetilde{\Pi}^{SWPTN}(\alpha, \beta)$ admits the representation
\begin{align}\label{eq:payoff}
\widetilde{\Pi}^{SWPTN}(\alpha,\beta)=NB_t\sum_{j=1}^{2\cN}w_{i,j}\mathbb{E}^{}\left[e^{v_0(T_l)+v_{i,j}(T_l)+\phi(T_l-T, u_{i,j}(T_l)+u_0(T_l))+\langle \psi(T_l-T,u_{i,j}(T_l)+u_0(T_l)), X_T\rangle  }\ind_\cG\biggr|\cF_t\right].
\end{align}
Recall from Definition~\ref{def:aff_model} that $u_i(T)+u_0(T)\in\mfU_T$, so that $\phi(T_{l}-T,u_{i,j}(T_l)+u_0(T_l))$ and $\psi(T_{l}-T,u_{i,j}(T_l)+u_0(T_l))$ are well-defined.  As shown in the following proposition, the quantity $\widetilde{\Pi}^{SWPTN}(\alpha,\beta)$ can be explicitly computed, analogously to Proposition \ref{prop:capletBasicAffine}.

\begin{proposition}\label{prop:swaptionBasicAffine}
Assume $\psi(T_l-T,u_{i,j}(T_l)+u_0(T_l)) \in \mathfrak{U}_T$ $\forall j = 1,\ldots,2\cN$. Let $\epsilon\in\RR$ such that  $\psi(T_{l}-T,u_{i,j}(T_l)+u_0(T_l))+\epsilon\beta\in \mathfrak{U}_T$, $\forall j = 1,\ldots,2\cN$.
Then the lower bound in terms of $\alpha,\beta$ for a payer swaption with notional $N$, maturity $T$, payment dates $T_1,..., T_\cN$, with $T_{j+1}-T_j=\delta_i$ for all $j=1,\ldots,\cN-1$ is given by
\begin{align}
\begin{aligned}
&\widetilde{\Pi}^{SWPTN}(\alpha,\beta)=NB_t\sum_{j=1}^{2\cN}w_{i,j}\Big(R_\epsilon\Big.\\
&\quad\Big.+\frac{1}{\pi }\int_{0-\im\epsilon}^{\infty-\im\epsilon}\Re\left(\frac{e^{-\im\zeta\alpha}}{\im\zeta}\mathbb{E}\left[e^{v_0(T_l)+v_{i,j}(T_l)+\phi(T_l-T, u_{i,j}(T_l)+u_0(T_l))+\langle \psi(T_l-T,u_{i,j}(T_l)+u_0(T_l))+\im \beta \zeta, X_T\rangle}\biggr|\cF_t\right]\right)d\zeta\Big),  
\end{aligned}
\label{eq:swaptionApproxGeneral}
\end{align}
where $R_\epsilon$ is given by
\[
R_\epsilon= 
\begin{cases}
\mathbb{E}\left[e^{v_0(T_l)+v_{i,j}(T_l)+\phi(T_l-T, u_{i,j}(T_l)+u_0(T_l))+\langle \psi(T_l-T,u_{i,j}(T_l)+u_0(T_l)), X_T\rangle}\biggr|\cF_t\right], & \mbox{if } \epsilon<0, \\ 
    \frac{1}{2}\mathbb{E}\left[e^{v_0(T_l)+v_{i,j}(T_l)+\phi(T_l-T, u_{i,j}(T_l)+u_0(T_l))+\langle \psi(T_l-T,u_{i,j}(T_l)+u_0(T_l)), X_T\rangle}\biggr|\cF_t\right],&\mbox{if } \epsilon = 0,\\
    0,&\mbox{if } \epsilon >0.
 \end{cases}
\]
\end{proposition}

\begin{proof}
The claim is a direct consequence of \cite[Theorem~5.1]{lee2004}, noting that each summand appearing in \eqref{eq:payoff} corresponds  in the notation of \cite{lee2004} to the case $G=G_3$, $b_1=\psi(T_l-T,u_{i,j}(T_l)+u_0(T_l))$, $b_0=\beta$ and $k=\alpha$. 
\end{proof}

Proposition \ref{prop:swaptionBasicAffine} gives a general lower bound for the price of a swaption, parameterized in terms of $(\alpha,\beta)$. 
These parameters should be determined in such a way that the lower bound becomes as tight as possible, while at the same time ensuring the finiteness of suitable exponential moments of $X_T$.
As pointed out in \cite{cfg15}, to which we refer for more details on the numerical implementation, the values of $(\alpha,\beta)$ can be chosen in two ways:  

\begin{enumerate}
\item 
by maximizing \eqref{eq:swaptionApproxGeneral} with respect to $\alpha, \beta$, thus providing the lower bound
\[
{\Pi}^{SWPTN}_{LB}:=\max_{\alpha,\beta}\widetilde{\Pi}^{SWPTN}(\alpha, \beta).
\]
Note that this solution can be computationally demanding, especially for highly dimensional models. Moreover, for a given choice of $\epsilon$, the optimization procedure should be constrained in order to ensure the finiteness of joint exponential moments of $X_T$.
\item 
by considering hyperplane-like approximations and predetermining the best possible values of $\alpha, \beta$ (the well-known Singleton-Umantsev approximation constitutes an example in this sense, see \cite{sinum2002} and compare also with \cite{kim14})\footnote{We also want to mention that, in the recent paper \cite{GPSS14}, the performance of the Singleton-Umantsev approximation has been empirically tested and compared to Monte-Carlo simulations in the context of a multi-curve model.}.
\end{enumerate}

\section{A tractable specification based on Wishart processes}	\label{sec:Wishart}

For notational simplicity we consider here the case of one single tenor $\delta$ and suppose that the driving process $X$ is a Wishart process on $(\Omega, \mathcal{F},(\mathcal{F}_t)_{0\leq t\leq \T}, \QQ)$ with state space $\mathbb{S}_d^+$ of the form 
\begin{align}\label{eq:Wish} 
dX_t=\left(\kappa Q^{\top}Q+MX_t+X_tM^\top\right)dt+\sqrt{X_t}dW_tQ+Q^\top dW^\top_t\sqrt{X_t}, \quad X_0=x_0,
\end{align}
where $W$ is $d \times d$ matrix of Brownian motions, $\kappa \geq d-1$ and $Q,M$ are $d\times d$ matrices.
The particular appealing feature of Wishart processes is stochastic correlation between the nonnegative diagonal elements of the matrix. On the classical canonical state space $\mathbb{R}^n \times \mathbb{R}^m_+$ 
this property cannot be achieved for positive factors. This possibility is however a crucial ingredient when it comes to modeling spreads which are highly correlated.

We consider an affine short rate multi-curve model as of Section~\ref{sec:short_rate} with short rate $r(t)=\ell(t)+\langle \lambda, X_t\rangle$
and spread $\log S^{\delta}(t,t)=c(t)+\langle \gamma, X_t\rangle $, where $c(t) \in \mathbb{R}_+$ and $\gamma \in \mathbb{S}_d^+$ to guarantee positivity of the spreads. Recall that the scalar product $\langle \cdot, \cdot \rangle$ is here the trace.

\begin{remark}
By choosing $\lambda$ and $\gamma$ to be diagonal matrices, the above model represents a natural extension of the classical CIR model to the multi-curve setting since the diagonal elements of a Wishart process are stochastically correlated CIR processes. 
\end{remark}

The goal of this section is to study the pricing of caps for this particular model.  As shown in Section~\ref{sec:option}, the price of a caplet with unitary notional can be computed via
\begin{equation}\label{eq:capletpriceWish1}
\begin{split}
\Pi^{CPLT}(t; T,T+\delta,K,1)&=S^{\delta}(t,T)B(t,T)\widetilde{\mathbb{Q}}\left[S^{\delta}(T,T)\geq (1+\delta K)B(T,T+\delta)\, | \, \mathcal{F}_t\right]\\
&\quad -(1+\delta K)B(t,T+\delta)\mathbb{Q}^{T+\delta}\left[S^{\delta}(T,T)\geq (1+\delta K)B(T,T+\delta)\, |\,  \mathcal{F}_t\right],
\end{split}
\end{equation}
where the probability measure $\widetilde{\mathbb{Q}}\sim\QQ$ is defined via
\[
\frac{d\widetilde{\mathbb{Q}}}{d \mathbb{Q}}:=\frac{S^{\delta}(T,T) B(T,T)}{B_T S^{\delta}(0,T) B(0,T)}.
\]
In the case of the above introduced model, it can be easily shown that the process $X$ follows under both measures $\widetilde{\mathbb{Q}}$ and $\mathbb{Q}^{T+\delta}$ a non-central Wishart distribution with time dependent parameters, which is stated in Lemma~\ref{lem:newWish} below.
As the density is (up to the solution of ODEs) explicitly known, this allows obtaining (semi-)analytical pricing formulas for caplets, similarly as in the CIR model.
As a preliminary, let us introduce the following definition, in line with~\cite[Definition A.4]{KK:14}.

\begin{definition}
Suppose that $\kappa \geq d-1$, $\Sigma \in \mathbb{S}_d^+$ and $\Theta$ is a $d \times d$ matrix such that $\Sigma\Theta$
is symmetric positive semidefinite. A symmetric positive definite random matrix $U$ is said to be \emph{non-centrally Wishart distributed}
with $\kappa$ degrees of freedom, covariance matrix $\Sigma$, and matrix of non-centrality parameter $\Theta$, if its Laplace transform satisfies
\[
\mathbb{E}[e^{-\langle u, U\rangle}]=\det(I_d-2 u\Sigma)^{-\frac{\delta}{2}}e^{- \langle u(I_d+2u\Sigma)^{-1},\Sigma\Theta\rangle}.
\]
In this case, we write $U ~\sim \mathcal{W}_d(\kappa, \Sigma, \Theta)$.
\end{definition}

\begin{lemma}\label{lem:newWish}
Let $X$ be a Wishart process of the form~\eqref{eq:Wish} under $\mathbb{Q}$.
\begin{enumerate}
\item 
Under $\widetilde{\mathbb{Q}}$, $X_T$ has a non-central Wishart distribution
\[
\mathcal{W}_{d}(\kappa, \widetilde{V}(0), \widetilde{V}(0)^{-1}\widetilde{\Psi}(0)^{\top} x \widetilde{\Psi}(0)),
\]
where $\widetilde{V}(t)$ and $\widetilde{\Psi}(t)$ are solutions of the following system of ordinary differential equations
\begin{equation}
\label{eq:ODE1}
\begin{aligned}
\partial_t \widetilde{\Psi}(t)&= -\left(M^{\top}+2 Q^{\top}Q\tilde{\psi}\left(T-t,\gamma,-\lambda \right)\right)\widetilde{\Psi}(t), &\quad \widetilde{\Psi}(T)=I_2,\\
\partial_t \widetilde{V}(t)&=-\widetilde{\Psi}(t)^{\top}Q^{\top}Q\widetilde{\Psi}(t), &\quad \widetilde{V}(T)=0,
\end{aligned}
\end{equation}
where $\tilde{\phi}$ and $\tilde{\psi}$ denote the characteristic exponents of the process $Y=(X, \int_0^{\cdot} X_s ds)$.
The solution of~\eqref{eq:ODE1} is explicitly given by
\begin{align*}
\widetilde{\Psi}(0)&=\exp\left(\int_0^T \left(M^{\top}+2 Q^{\top}Q\tilde{\psi}\left(T-t,\gamma, -\lambda\right)\right)dt\right) I_2\\
\widetilde{V}(0)&=\int_0^T \exp\left(\int_t^T A_s ds\right)Q^{\top}Q\exp\left(\int_t^T A^{\top}_s ds\right) dt
\end{align*}
with $A_s:=M+2 \tilde{\psi}\left(T-s, \gamma,-\lambda \right)Q^{\top}Q$.
\item 
Under $\mathbb{Q}^{T+\delta}$, $X_T$ has a non-central Wishart distribution
$\mathcal{W}_d(\kappa, V(0), V(0)^{-1}\Psi(0)^{\top} x \Psi(0))$, 
where $V(t)$ and $\Psi(t)$ are solutions of the following system of ordinary differential equations
\begin{equation}
\label{eq:ODE2}
\begin{aligned}
\partial_t \Psi(t)&= -(M^{\top}+2 Q^{\top}Q\tilde{\psi}(T-t), 0,-\lambda)\Psi(t), &\quad \Psi(T)=I_2,\\
\partial_t V(t)&=-\Psi(t)^{\top}Q^{\top}Q\Psi(t), &\quad V(T)=0,
\end{aligned}
\end{equation}
where $\tilde{\phi}$ and $\tilde{\psi}$  denote the characteristic exponents of the process $Y=(X, \int_0^{\cdot} X_s ds)$. The solution of~\eqref{eq:ODE2} is explicitly given as above.
\end{enumerate}
\end{lemma}

\begin{proof}
Concerning (i), note that the density process $(N_t)_{0\leq t\leq T}$ of
$\frac{d\widetilde{\mathbb{Q}}}{d \mathbb{Q}}$ is given by
\begin{align*}
N_t&:=\mathbb{E}\left[\frac{d\widetilde{\mathbb{Q}}}{d \mathbb{Q}}\Big| \mathcal{F}_t\right]=
 \frac{1}{S^{\delta}(0,T) B(0,T)}\mathbb{E}\left[\frac{S^{\delta}(T,T) B(T,T)}{B_T} \Big| \mathcal{F}_t\right]\\
 &= \frac{1}{S^{\delta}(0,T) B(0,T)}\mathbb{E}\left[e^{c(T)+\langle \gamma, X_{T}\rangle-\int_0^T \ell(s)ds-\langle\lambda, \int_0^T X_{s}ds\rangle } \Big| \mathcal{F}_t\right]\\
 &= \frac{1}{S^{\delta}(0,T) B(0,T)}\exp\left(c(T)-\int_0^T \ell(s)ds+\tilde{\phi}\left(T-t, \gamma,-\lambda \right)+\left\langle\tilde{\psi}\left(T-t, \gamma, -\lambda\right), X_t\right\rangle-\langle \lambda, \int_0^t X_{s}ds\rangle\right),
\end{align*}
where $\tilde{\phi}$ and $\tilde \psi$ denote the characteristic exponents of $Y=(X, \int_0^{\cdot} X_s ds)$.
Note that the diffusion part of $N_t$ is given by 
\[
2\int_0^tN_s  \left\langle Q \tilde{\psi}\left(T-s, \gamma, -\lambda \right)\sqrt{X}_s, dW_s \right\rangle,
\]
so that we can write 
\[
N_t=\mathcal{E}\left(2\int_0^{\cdot}  \left\langle Q \tilde{\psi}\left(T-s, \gamma,-\lambda \right)\sqrt{X}_s, dW_s \right\rangle\right)_t.
\] 
By Girsanov's theorem, under the measure $\widetilde{\QQ}$, the linear drift of $X$ changes to
\[
M+2 \tilde{\psi}\left(T-t, \gamma, -\lambda \right)Q^{\top}Q,
\]
so that $X$ becomes a Wishart process with time-varying linear drift under $\widetilde{\mathbb{Q}}$. According to~\cite[Proposition A.6]{KK:14}, $X_T$ has a noncentral Wishart distribution $\mathcal{W}_d(\kappa, \widetilde{V}(0), \widetilde{V}(0)^{-1}\widetilde{\Psi}(0)^{\top} x \widetilde{\Psi}(0))$, 
where $\widetilde{V}(t)$ and $\widetilde{\Psi}(t)$ are solutions of~\eqref{eq:ODE1}.
Concerning (ii), we have for the density process of $\frac{d\mathbb{Q}^{T+\delta}}{d \mathbb{Q}}$
\begin{align*}
\mathbb{E}\left[\frac{d\mathbb{Q}^{T+\delta}}{d \mathbb{Q}}\Big| \mathcal{F}_t\right]&=\frac{B(t,T)}{B_t B(0,T)}\\
 &= \frac{1}{B(0,T)}\exp\left(-\int_0^T \ell(u)du+\tilde{\phi}\left(T-t, 0, -\lambda \right)+\left\langle\tilde{\psi}\left(T-t, 0,-\lambda\right), X_t\right\rangle-\langle \lambda, \int_0^t X_{s}ds\rangle \right).
\end{align*}
The assertion then follows similarly as for $\widetilde{\mathbb{Q}}$.
\end{proof}

\subsection{Computing certain probabilities of linear functionals in non-central Wishart distributions}

In view of~\eqref{eq:capletpriceWish1}, we focus here on the computation of 
\[
\widetilde{\mathbb{Q}}\left[S^{\delta}(T,T)\geq (1+\delta K)B(T,T+\delta)\right]
\qquad\text{and}\qquad
\mathbb{Q}^{T+\delta}\left[S^{\delta}(T,T)\geq (1+\delta K)B(T,T+\delta)\right].
\]
By the specification of our model, the first of the two quantities above becomes
\[
\widetilde{\mathbb{Q}}\left[\langle\gamma-\tilde{\psi}(\delta,0,-\lambda), X_{T}\rangle\geq \log(1+\delta K) -\int_T^{T+\delta}\ell(u) du+ \tilde{\phi}(\delta,0,-\lambda)-c(T)\right]
\]
and similarly for the $\mathbb{Q}^{T+\delta}$-probability.
It thus amounts to compute expressions of the type 
\[
\widetilde{\mathbb{Q}}\left[\langle A, X_T \rangle \geq C\right],
\]
for some matrix $A \in \mathbb{S}_d$ (note that $\tilde{\psi}(\delta,0,-\lambda)$ can be symmetrized  to lie in $\mathbb{S}_d$) and some constant $C$. 
The following proposition relates linear combination of elements of non-centrally Wishart distributed matrices with $\chi^2$-distributed random variables.

\begin{proposition}\label{prop:chisquare}
Let $X ~\sim W_d(\kappa, \Sigma, \Sigma^{-1}x)$ with $\Sigma \in \mathbb{S}_d^{++}$. Then 
\[
\langle A,X \rangle \sim \sum_{i=1}^d \lambda_i V_i,
\]
where $\lambda_i$ are the eigenvalues of $\sqrt{\Sigma} A \sqrt{\Sigma}=O\Lambda O^{\top}$ and $V_{i}, i \in \{1,\ldots,d\}$ are independent random variables with $V_i \sim \chi^2(\kappa, y_{ii})$, where
$
y=O^{\top}\Sigma^{-\frac{1}{2}}x\Sigma^{-\frac{1}{2}}O.
$
\end{proposition}

\begin{proof}
By~\cite[Proposition 6]{AA:13}, $X$ has the same distribution as  $\sqrt{\Sigma} Z \sqrt{\Sigma}$ where $Z \sim   \mathcal{W}_d(\kappa, I_d, \Sigma^{-\frac{1}{2}}x\Sigma^{-\frac{1}{2}} )$. Therefore 
\begin{align*}
\langle A,X\rangle &~\sim \langle A, \sqrt{\Sigma} Z\sqrt{\Sigma} \rangle= \langle \sqrt{\Sigma} A \sqrt{\Sigma}, Z \rangle
=\langle O\Lambda O^{\top}, Z \rangle
=\langle \Lambda, O^{\top} Z O\rangle
=:\langle \Lambda, Y\rangle
=\sum_{i=1}^d \lambda_i Y_{ii}
\end{align*}
where $Y:=O^{\top} Z O ~\sim  W_d(\kappa, I_d, O^{\top}\Sigma^{-\frac{1}{2}}x\Sigma^{-\frac{1}{2}}O )=W_d(\kappa, I_d, y)$. 
Let us now compute the moment generating function of $\sum_{i=1}^d \lambda_i Y_{ii}$, which is given by (see, e.g.,~\cite[Proposition A.5]{KK:14})
\[
\mathbb{E}\left[e^{ u \sum_{i=1}^d \lambda_i Y_{ii}}\right]=\prod_{i=1}^d(1-2u\lambda_i)^{-\frac{\kappa}{2}}e^{\frac{u\lambda_iy_{ii}}{1-2u\lambda_i}}.
\]
However, this corresponds to the moment generating function of $\sum_{i=1}^d\lambda_iV_i$, where $V_i ~\sim\chi^2(\kappa, y_{ii})$ are independent random variables. Indeed, it holds that 
\[
\mathbb{E}\left[e^{u ( \sum_{i=1}^d\lambda_i V_{i})}\right]=\prod_{i=1}^d\mathbb{E}\left[e^{u \lambda_i V_{i}}\right]
=\prod_{i=1}^d(1-2u\lambda_i)^{-\frac{\kappa}{2}}e^{\frac{u\lambda_iy_{ii}}{1-2\lambda_iu}}.
\]
This proves that $\sum_{i=1}^d \lambda_i Y_{ii} \sim \sum_{i=1}^d \lambda_i V_{i}$. Since $\langle A, X_T \rangle~\sim \sum_{i=1}^d \lambda_i Y_{ii}$, the assertion is proved.
\end{proof}

\begin{corollary}\label{cor:Wish}
Let $A \in \mathbb{S}_d$ and $X$ a Wishart process of form~\eqref{eq:Wish}. Then under $\mathbb{Q}^{T+\delta}$ it holds that 
\[
\langle A, X_T\rangle ~\sim \sum_{i=1}^d \lambda_{i,T}V_{i,T},
\]
where $\lambda_{i,T}$ are the eigenvalues of $\sqrt{V(0)}A \sqrt{V(0)}=O\Lambda_T O^{\top}$ and $V_{i,T}, i\in \{1,\ldots,d\}$ are independent random variables with $V_{i,T} \sim \chi^2(\kappa, y_{ii,T})$, where
\[
y_T=(O^{\top}V(0)^{-\frac{1}{2}}\Psi^{\top}(0)\Psi(0)V(0)^{-\frac{1}{2}}O)
\]
and $V(0)$ and $\Psi(0)$ are given in Lemma~\ref{lem:newWish}. The same assertion holds for $\mathbb{\widetilde{Q}}$ with $V(0)$ and $\Psi(0)$ replaced by $\widetilde{V}(0)$ and $\widetilde{\Psi}(0)$.
\end{corollary}

\begin{proof}
The assertion is a direct consequence of Lemma~\ref{lem:newWish} and Proposition~\ref{prop:chisquare}.
\end{proof}

\subsection{A closed-form expression for the price of caplet}

Finally, by relying on the above results, we are ready to give a (semi-)analytical formula for the price of a caplet in the above Wishart model.
\begin{theorem}
Let $X$ be a Wishart process of the form~\eqref{eq:Wish}. %with $M_{21}=0$. 
Consider an affine short rate multi-curve model with 
\begin{itemize}
\item 
OIS short rate $r_t=\ell(t)+\langle \lambda, X_t\rangle$, for all $t \in [0, \T]$ and 
\item logarithmic multiplicative spreads $\log S^{\delta}(t,t)=c(t)+\langle \gamma,  X_{t} \rangle$, for all $t \in [0, \T]$. 
\end{itemize}
Then the price of a caplet at date $0$ with unitary notional,  reset date $T$ and payoff $\delta(L_T(T,T+\delta)-K)^+$ at the settlement date $T+\delta$, is given by
\begin{equation}\label{eq:capletpriceWish2}
\Pi^{CPLT}(0; T,T+\delta,K,1)=S^{\delta}(0,T)B(0,T)\bigl(1- \widetilde{F}_T(C_{T,K})\bigr)%\int_{C}^{\infty} \widetilde{f}(x)dx
-(1+\delta K)B(0,T+\delta)\bigl(1-F_T(C_{T,K})\bigr),
\end{equation}
where 
\begin{itemize}
\item the constant $C_{T,K}$ is given by
\begin{align}\label{eq:constC}
C_{T,K}= \log(1+\delta K) -\int_T^{T+\delta}\ell(u) du+ \tilde{\phi}(\delta,0,-\lambda)-c(T),
\end{align} 
with $\widetilde{\phi}$ being the constant part in the characteristic exponent of $(X, \int_0^{\cdot} X_s ds)$, and
\item 
$\widetilde{F}_T$ and $F_T$ denote the cumulative distribution functions of a weighted sum of  non-centrally $\chi^2$-distributed random variables corresponding to  $\sum_{i=1}^d \widetilde{\lambda}_{i,T}\widetilde{V}_{i,T}$ and $\sum_{i=1}^d\lambda_{i,T} V_{i,T} $ as of Corollary~\ref{cor:Wish} for $\widetilde{\mathbb{Q}}$ and $\mathbb{Q}^{T+\delta}$ with $A=\gamma-\tilde{\psi}(\delta,0,-\lambda)$ and 
$\widetilde{\psi}$ the constant part in the characteristic exponent of $(X, \int_0^{\cdot} X_s ds)$.
\end{itemize}
\end{theorem}

\begin{proof}
The assertion is a consequence of equation~\eqref{eq:capletpriceWish} together with Corollary~\ref{cor:Wish}.
\end{proof}

\begin{remark}
In order to practically implement the pricing formula~\eqref{eq:capletpriceWish2}, the following steps are necessary:
\begin{itemize}
\item
compute $\widetilde{\psi}$ and integrate it over the interval $[0,T]$ to obtain $V(0), \widetilde{V}(0), \Psi(0), \widetilde{\Psi}(0)$;
\item compute the eigenvalues and eigenvectors of $\sqrt{\widetilde{V}(0)}A \sqrt{\widetilde{V}(0)}$ and $\sqrt{V(0)}A \sqrt{V(0)}$ with $A=\gamma-\tilde{\psi}(\delta,0,-\lambda)$ to get the weights $\lambda_T, \widetilde{\lambda}_T$ and the non-centrality parameters $y_T, \widetilde{y}_T$;
\item compute $C_{T,K}$ as given in~\eqref{eq:constC} and the distribution function of a positive weighted sum of two independent non-centrally $\chi^2$-distributed random variables, e.g., via a Laguerre series expansions (see~\cite{CL:05}) of the form
\[
F_T(C)=F_T(C_{T,K},\kappa, y_T, \lambda_T)=g(\kappa, C_{T,K})\sum_{j =0 }^n \alpha_j(\kappa,y_T,\lambda_T)L_j^{(\kappa)}(qC_{T,K}),
\]
where $g$ is a function depending on $\kappa, C_{T,K}$ and $\alpha_j$ are the coefficients (depending  on $\kappa,y_T,\lambda_T$) of the Laguerre polynomials $L_j^{(\kappa)}$ (evaluated at $qC_{T,K}$ where $q$ denotes a constant). For an alternative computation based on Fourier inversion we refer to \cite{J:96}.
\end{itemize}
Note that the weights $\lambda_T, \widetilde{\lambda}_T$ and the non-centrality parameters $y_T,\widetilde{y}_T$ depend on the maturity (the degrees of freedom $\kappa$ are constant) while the arguments $C_{T,K}$ in the distribution function depend additionally on the strike. In the case where $\ell$ and $c$ are constants, $C$ only depends on $K$ and we write only $C_K$. In this case
the matrix $(F_T(C_{K}))_{T \in \{T_1, \ldots, T_m\}, K \in \{K_1, \ldots, K_l\}}$ for different maturities and strikes needed for calibration purposes can be obtained by a matrix product $UV$ where $U \in \mathbb{R}^{m \times n}$ and $V \in \mathbb{R}^{n\times l}$  are defined by 
 \[
 U_{ij}=\alpha_j(\kappa,y_{T_i},\lambda_{T_i}) 
 \quad\textrm{ and }\quad 
 V_{ij}= L_i^{(\kappa)}(qC_{K_j})g(\kappa, C_{K_j}).
 \]
Provided the initial term structure of spreads and bonds is known, this procedure then gives the prices for all maturities and strike rates.
\end{remark}

\section{Calibration analysis} \label{sec:aff_examples}
%Chage these commands in order to change calibration folder/date
\newcommand{\dateFolder}{08022011}
\newcommand{\calibrationDate}{dataset\dateFolder}
\newcommand{\theDate}{August $2^{nd}$, 2011}

In this section, we discuss the calibration of two simple specifications of our general framework to cap/floor market data. This section aims at illustrating the practical feasibility of the proposed approach and it is not meant to suggest a particular specification of the framework. 
The first specification is based on a CIR-Gamma model, while the second specification is driven by a Wishart process, as presented in Section~\ref{sec:Wishart}. In particular, this represents the first instance of calibration of a Wishart short rate model in the multi-curve framework. Moreover, we complement our results with a parameter stability analysis in the case of the CIR-Gamma model, showing a satisfactory degree of stability.
We refer to \cite{Bormetti2015,CGNS:13} for calibration results based on swaption data and  to \cite{mopa10} for a calibration approach which relies on ATM European swaption and cap quotes.

\subsection{Market data} 

Let us start by briefly describing our market data sample. We initially consider a fixed trading date, namely \theDate. The data sample consists of a cross section of market quotes (corresponding to perfectly collateralized transactions) of linear and non-linear interest rate derivatives.
As far as linear products are concerned, we consider market data for overnight indexed swaps and interest rate swaps.
On the basis of these market quotes, we construct the OIS discount curve $T \mapsto B(0,T)$ and the forward curves $T\mapsto L_0(T,T+\delta_i)$, for $\delta_1= 3M$ and $\delta_2 = 6M$. 
This has been performed by relying on the Finmath Java library (see \cite{finmath}). 

\begin{figure}[ht]
  \centering
  \subfloat{\label{fig:discountCurve\calibrationDate}\includegraphics[scale=0.35]{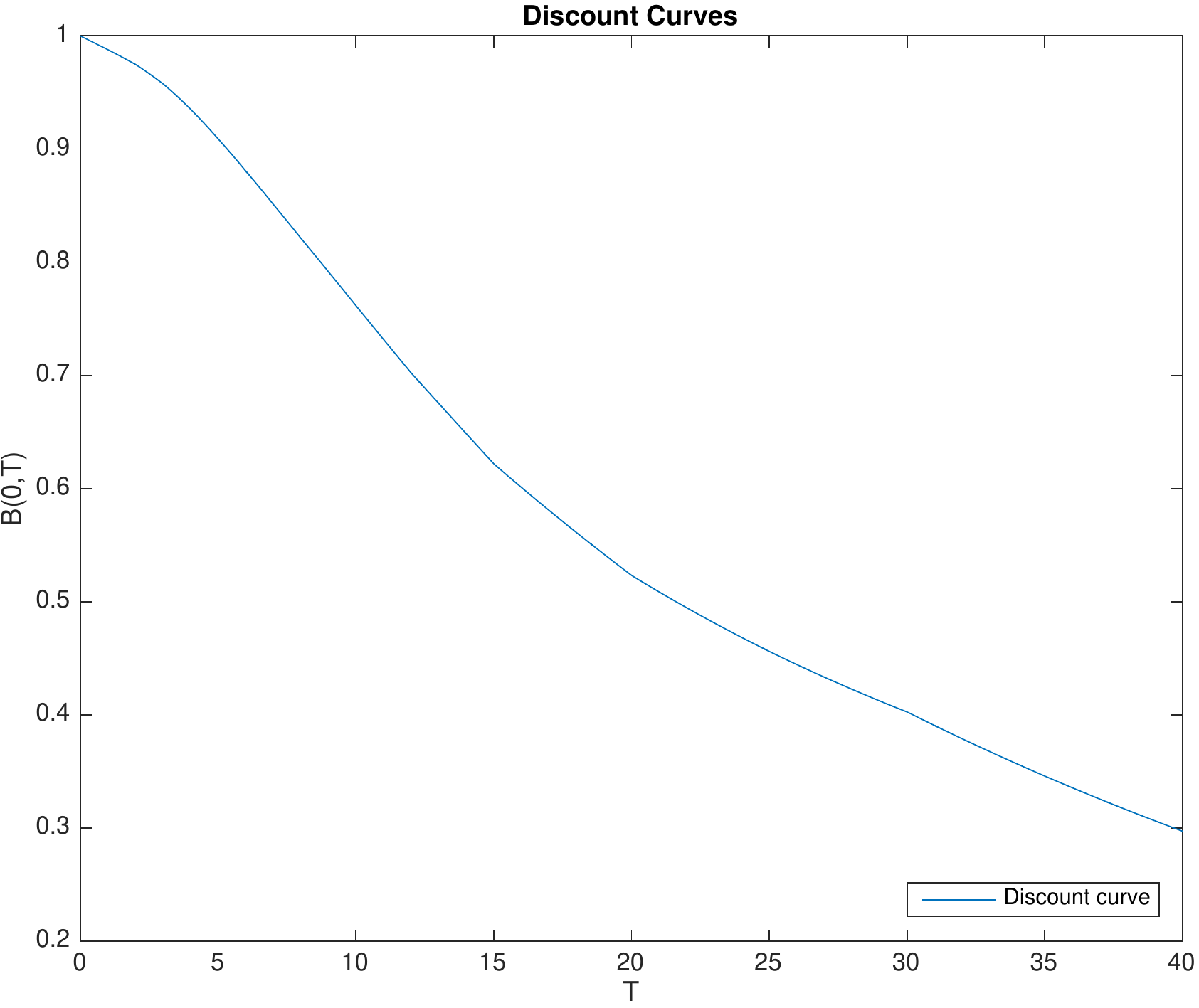}} \                                          \subfloat{\label{fig:forwardCurve\calibrationDate}\includegraphics[scale=0.35]{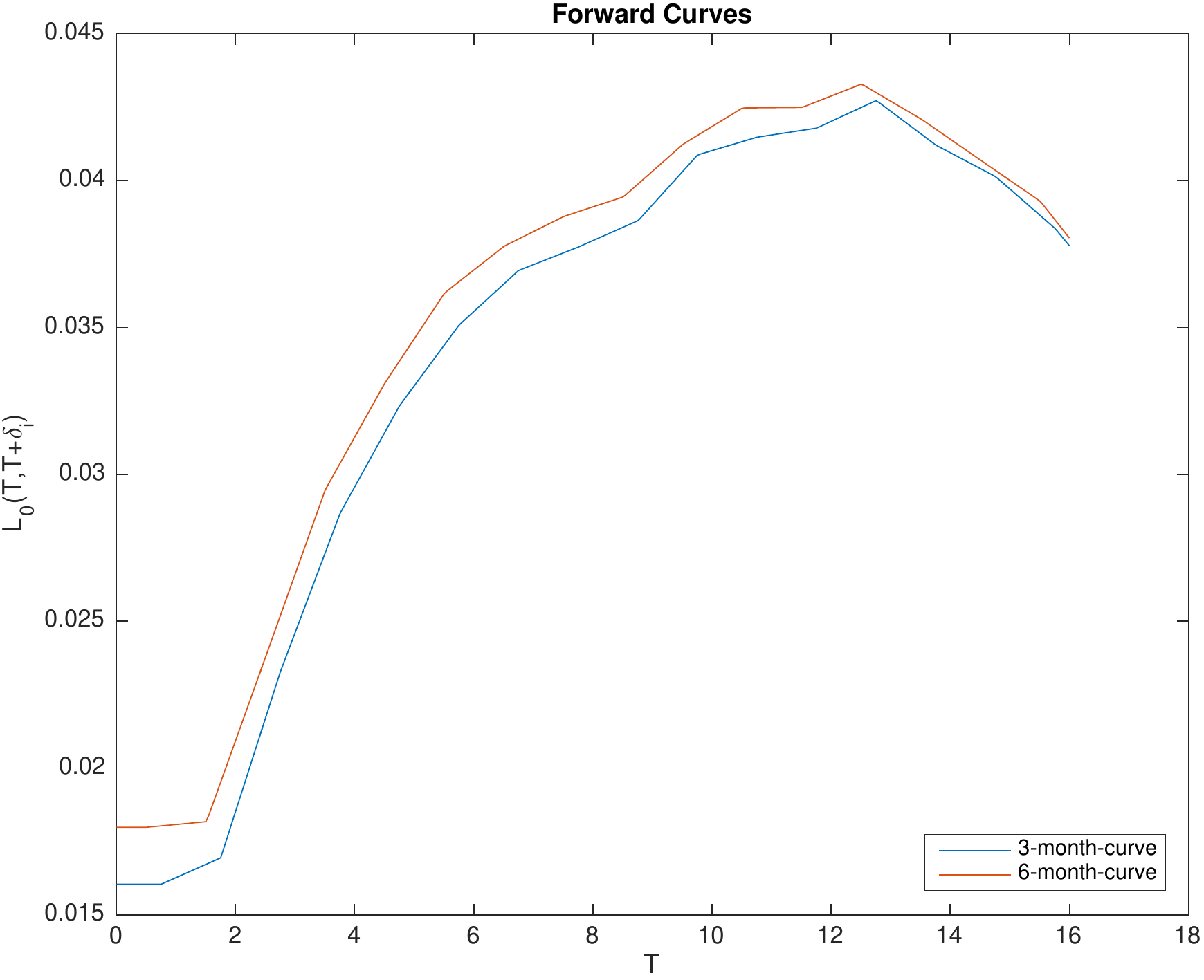}}
\caption{Discount and forward curves as of \theDate. \label{fig:curves\calibrationDate}}
\end{figure}

Concerning non-linear interest rate products, the market practice consists in posting a cap/floor composite surface: when the strike price is below the at-the-money (ATM) level, quotes refer to out-of-the-money (OTM) floors, whereas if the strike is above ATM, quotes refer to OTM caps. 
In this way, the whole surface is constructed from liquidly traded OTM options.
By relying on the put-call parity between caps and floors, we can however treat all implied volatilities as cap volatilities. To reduce the complexity of the calibration, it is convenient to construct a caplet implied volatility surface, suitably bootstrapped from cap implied volatility data.
The surface refers to strike prices ranging between $0.75\%$ and $6\%$ and  maturities between $6$ months and $10$ years. In the market, caps with maturity larger than two years are indexed to the $6$-month forward rate while those with lower expiry are linked to the $3$-month curve. 
A normal implied volatility is obtained by numerically searching for the value of $\sigma_N$ (where $N$ stands for normal) such that the Bachelier pricing formula for a caplet
\begin{align*}
\Pi^{CPLT}(t; T_{i-1},T_{i},K,1)&=B(t,T_i)\delta\Excond{\QQ^{T_{i}}}{\left(L_{T_{i-1}}(T_{i-1},T_{i})-K\right)^+}{\cF_t}\\
&=B(t,T_i)\delta\sigma_N\sqrt{T_{i-1}-t}\left(\frac{1}{\sqrt{2\pi}}e^{-\frac{z^2}{2}}+zN(z)\right),
\end{align*}
where
\[
N(x)=\frac{1}{\sqrt{2\pi}}\int_{-\infty}^xe^{-\frac{y^2}{2}}dy
\qquad\text{and}\qquad
z = \frac{L_t(T_{i-1},T_{i})-K}{\sigma_N\sqrt{T_{i-1}-t}},
\]
best fits the market price of a given caplet.

\subsection{Implementation details}

We now provide a detailed description of the implementation of Proposition~\ref{prop:capletBasicAffine}, by means of the FFT algorithm (see \cite{cm99}). In the present discussion, we express the integration variable as $\zeta=z-\im\epsilon$. Let also $k:=\log \bar{K}$ and
\[
\cI(z):=\frac{\varphi_{\mathcal{Y}_{T}}(z-\im(\epsilon+1))}{-(z-\im\epsilon)(z-\im(\epsilon+1))}.
\]
The integral term appearing in \eqref{capletprice} can then be written as
\begin{align*}
IT(k):=\frac{e^{-\epsilon k}}{\pi}\int_0^\infty\Re\left(e^{-\im zk}\cI(z)\right)dz.
\end{align*}
We perform a first approximation by introducing a trapezoidal rule of the form $z_{j_1}:=\eta(j_1-1), \ \eta >0, \ j_1=1,\ldots,N$, so that the effective upper limit of integration is given by $(N-1)\eta$. As we want to perform a simultaneous evaluation of the integral term for a grid of strike prices, we also introduce a grid for $k$ of the form $k_{j_2}:=-b+\eta^\star(j_2-1), \ \eta^\star>0, \ j_2=1,\ldots,N$, which gives a mesh covering the interval $[-b, b)$, with $b = 0.5N\eta^\star$. 
Since we want to apply the FFT algorithm, we need to impose the Nyquist condition, meaning that we set $\eta\eta^\star=2\pi/N$, thus introducing a tradeoff between the accuracy of the log-strike and integration grids. As suggested by \cite{cm99}, we introduce the weights of the Simpson rule, in order to obtain a satisfactory accuracy even for large values of $\eta$. In summary, the integral term (along a grid of log-strikes) is approximated via
\begin{align}
\label{eq:capletFFT}
IT(k_{j_2})\approx\frac{e^{-\epsilon k_{j_2}}}{\pi}\Re\sum_{j_1=1}^{N}e^{-\im \frac{2\pi}{N}(j_1-1)(j_2-1)}e^{\im z_{j_1}b}\cI(z_{j_1})\frac{\eta}{3}\left[3+(-1)^{j_1}+\delta_{j_1-1}\right],
\end{align}
where $\delta_n$ denotes a Kronecker delta function which is $1$ for $n=0$ and zero otherwise. Formula \eqref{eq:capletFFT} can be computed by a direct application of the FFT algorithm. In our analysis, we set $N=16384$ and $\eta=0.2$. 

For a given vector of model parameters $p$, belonging to the set of admissible parameters $\cP$, we compute caplet prices using the above methodology and convert them into model implied normal volatilities, that we denote by $\sigma^{imp}_{mod}(p)$. The aim of the calibration procedure is to solve
\begin{align*}
\min_{p\in\cP}\left\|\sigma^{imp}_{mkt}-\sigma^{imp}_{mod}(p)\right\|^2,
\end{align*}
where $\sigma^{imp}_{mkt}$ denotes the market-observed implied volatilities and $\|\cdot\|$ the Euclidean norm.

\subsection{Calibration results} 
In the following, we illustrate two candidate specifications along with their calibration results. Note that both models allow for a perfect fit to the observed term structures via a suitable choice of the functions $\ell$ and $c$ in line with Proposition~\ref{prop:deterministicShiftCharacterization}.

\subsubsection{CIR-Gamma model}	\label{sec:CIRGamma}
We first calibrate the following model consisting of a two-dimensional process $X=(X^1,X^2)$ of the form
\begin{align*}
X^1_t=X^1_0+ \int_0^t \left(b+\beta X^1_s\right)ds+\int_0^t \sigma\sqrt{X^1_s}dW_s, \qquad 
X^2_t=X^2_0+\int_0^t \int \xi \mu(d\xi),
\end{align*}
where  $b, \beta, \sigma \in \mathbb{R}$,  $W$ is a Brownian motion and $X^2$ is a Gamma process with compensating jump measure $\nu(d\xi)=mx^{-1}e^{-n\xi}d\xi$, with $m,n>0$.
The short rate is specified as $r_t=\ell(t)+\lambda X^1_t$ and the spreads are of the form $\log S^{\delta_i}(t,t)=c_i(t)+\gamma_{i}(X^1_t+X^2_t)$, $i=1,2$, with $\gamma_i,\lambda \in \mathbb{R}$ and $\ell,c:[0,\T]\rightarrow\RR$. 

Table~\ref{table:CIR-VGparams} reports the calibrated model parameters  while Figure~\ref{fig:calibrationErrors\calibrationDate} illustrates the quality of the fit in terms of prices and implied volatilities. More precisely, the left panel shows the squared error in price while the right panel illustrates the squared error in implied volatilities. 
Despite its simplicity, the model achieves a reasonably good fit to the observed market quotes. 
Looking at squared error in implied volatilities, we observe that the quality of the fit is lower for higher strikes and for the first maturity we considered (6 months). 
However, the quality of the fit is acceptable as can be seen by looking at the left panel of Figure~\ref{fig:calibrationErrors\calibrationDate}, which highlights that the squared error in price in the $6$-month maturity and deep OTM region is low.
Observe that the calibrated parameters satisfy $\gamma_1<\gamma_2$, reflecting the fact that $6$-month rates embed a higher degree of interbank risk with respect to $3$-month rates. This was achieved without imposing a priori constraints in the optimization algorithm.
Note also that the computational complexity of the procedure is at the same level of a calibration of a standard affine stochastic volatility model for equities.

\begin{table}[ht]
\centering
\begin{tabular}{|cc|cc|cc|}
\hline
\multicolumn{2}{|c|}{$X^1$}&\multicolumn{2}{c}{$X^2$}&\multicolumn{2}{|c|}{Parameters}\\
\hline
\hline
$b$ & $0.0630$ & $m$& $0.3651$ & $\lambda$& $0.0107$\\
$\beta $& $0.0033$ & $n$ & $1.8614$ & $\gamma_{1}$ & $0.0039$\\
$\sigma$ & $0.1479$& $X^2_0$ & $0.2386$ &$\gamma_{2}$ &$0.0128$\\
$X^1_0$& $0.4330$ & & & &\\
\hline
Resnorm &$0.0014$&&&&\\
\hline
\end{tabular}
\caption{Calibration result for the CIR-Gamma model. Resnorm represents the sum of squared distances between market and model implied volatilites.\label{table:CIR-VGparams}}
\end{table}

\begin{figure}[ht]
  \centering
  \subfloat{\label{fig:priceErrors\calibrationDate}\includegraphics[scale=0.40]{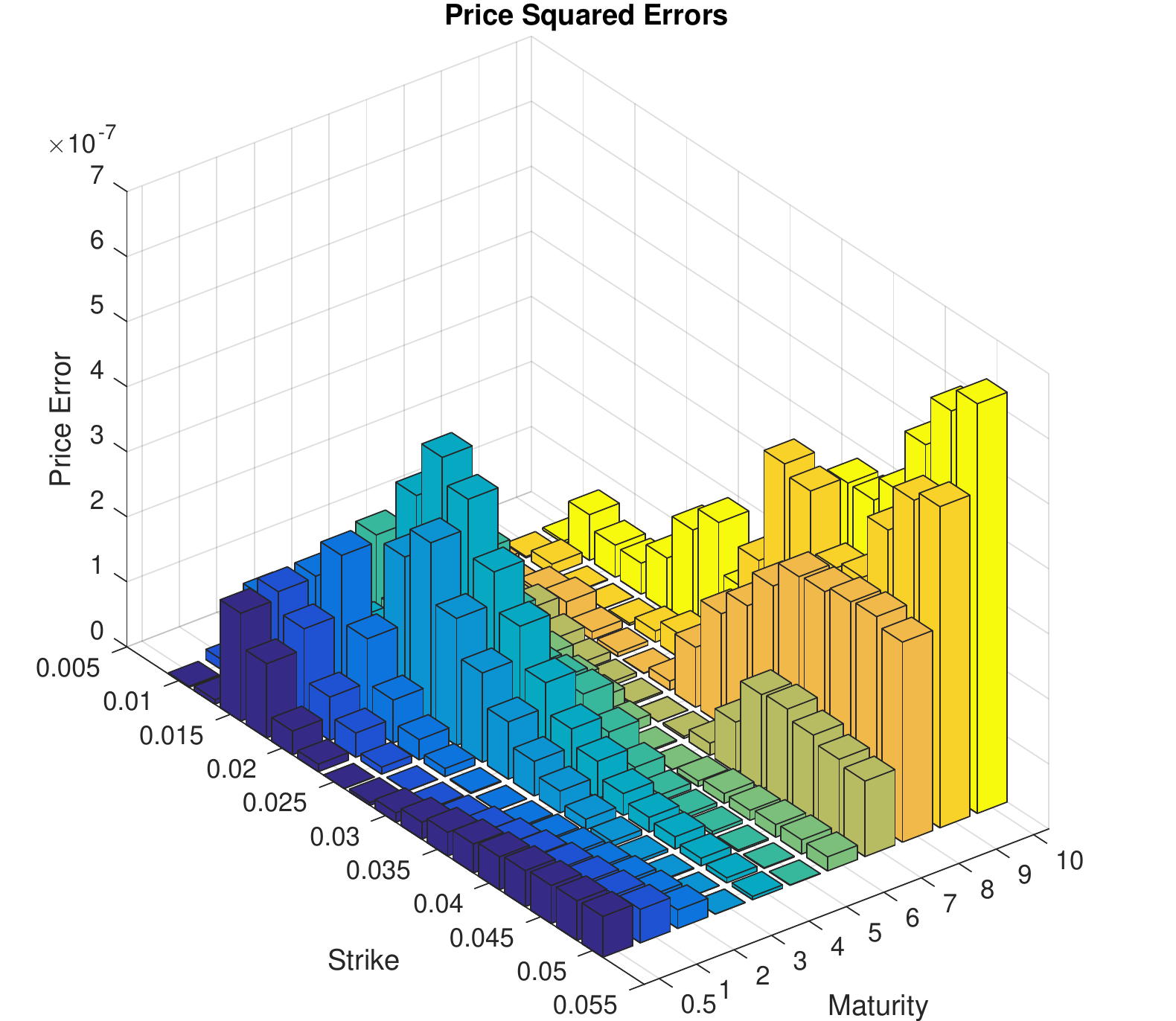}} \                                       \subfloat{\label{fig:volaErrors\calibrationDate}\includegraphics[scale=0.40]{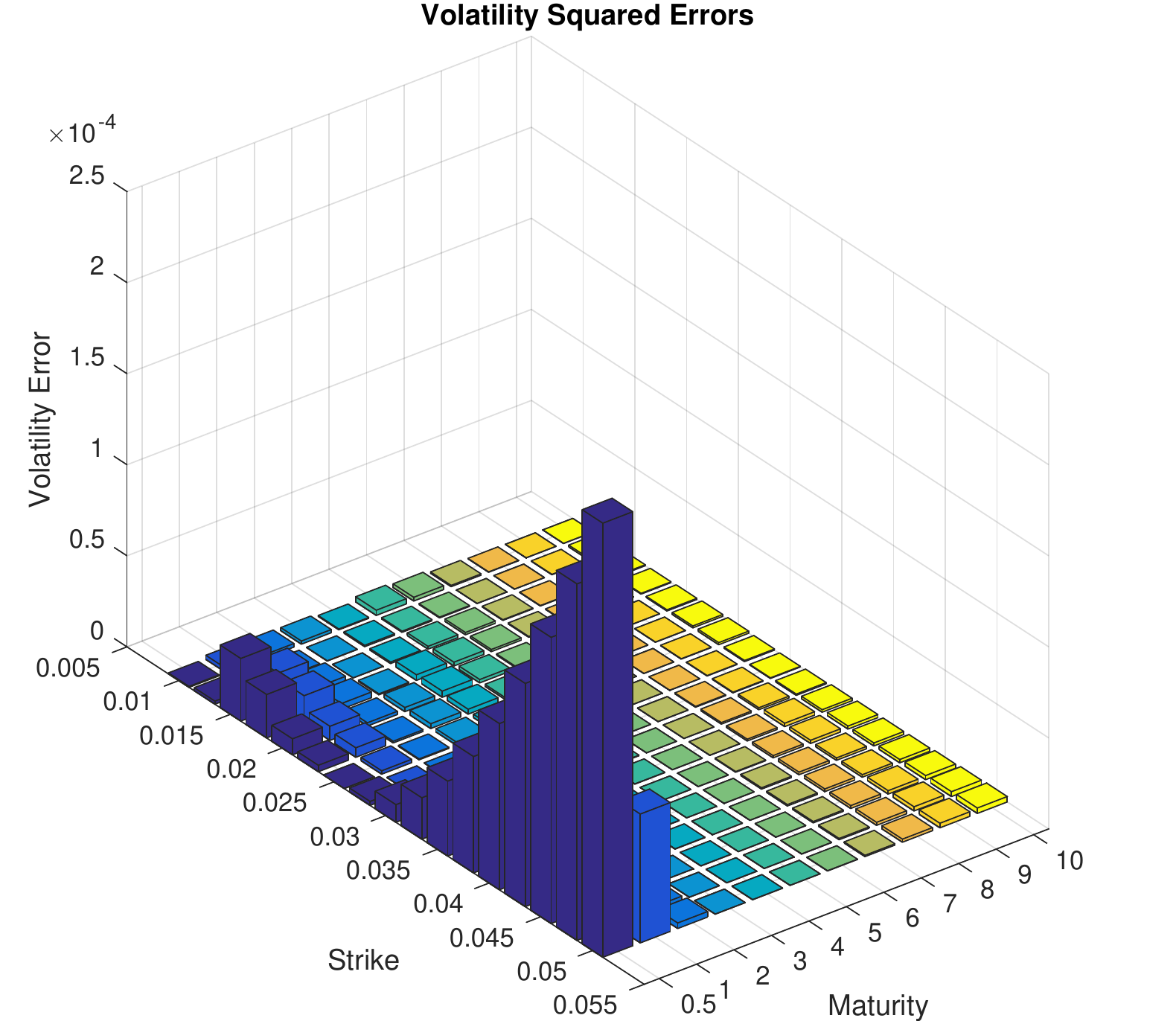}}\
\caption{Calibration residuals in terms of prices and implied volatilities for the CIR-Gamma model. Calibration date: \theDate. \label{fig:calibrationErrors\calibrationDate}}
\end{figure}

\subsubsection{Wishart-Gamma model}

We now present the results of the calibration to the same market data sample of a simple model driven by a Wishart process.
To the best of our knowledge, this represents the first example for the calibration of a Wishart-driven interest rate model to market data of non-linear products. Previous works such as \cite{Biagini2016,dfg13,gnoatto2012} were limited to the presentation of models and their properties.
Let $X=(X^1,X^2)$ be a process where $X^1$ is a Wishart process as defined in~\eqref{eq:Wish} with $d=2$ and $X^2$ is a Gamma process exactly as in the model considered in Section~\ref{sec:CIRGamma}. The short rate is specified as $r_t=\ell(t)+ \lambda \langle I_d, X^1_t\rangle $ and the spreads as $\log S^{\delta_i}(t,t)=c_i(t)+\gamma_i(\langle I_d, X^1_t \rangle +X^2_t)$, $i=1,2$, with $\gamma_i, \lambda \in \mathbb{R}$.
We call this specification {\em Wishart-Gamma model}. 
We report  in Table~\ref{table:Wishart-VGparams} the calibrated parameters, while Figure~\ref{fig:calibrationErrorsWis\calibrationDate}
illustrates the quality of the fit in terms of price and implied volatility errors. The magnitude of the calibration error is in line with the previous CIR-Gamma specification. Notice that, also in this case, the calibrated parameters satisfy $\gamma_1<\gamma_2$. We acknowledge however that the computational cost of this specification is higher due to the fact that one needs to solve matrix Riccati ODEs instead of classical ones.
Constraining some parameters or an application of the approach from Section~\ref{sec:Wishart} could represent solutions which are left for future research.

\begin{remark}
We point out that, since the caplet pricing formula of Proposition \ref{prop:capletBasicAffine} only depends on the specific form of the solutions to the system of generalized Riccati ODEs, the overall structure of the software implementation does not depend on the specific combination of processes/state spaces. 
\end{remark}

\begin{table}[ht]
\centering
\begin{tabular}{|cc|cc|cc|}
\hline
\multicolumn{2}{|c|}{$X^1$}&\multicolumn{2}{c}{$X^2$}&\multicolumn{2}{|c|}{Parameters}\\
\hline
\hline
$\kappa$ & $3.0626$ & $m$ & $0.3502$ & $\lambda$& $0.0021$ \\ 
$M$& $\left(\begin{array}{cc}
-0.4647 & -0.0218 \\ 
-0.0823 & 0.0110
\end{array} \right)$  & $n$ & $3.8926$ & $\gamma_{1}$& $0.0068$  \\
$Q$& $\left(\begin{array}{cc}
-0.0093 & 0.0201 \\ 
-0.0008 & 0.1019
\end{array} \right)$  & $X^2_0$& $2.7617$ &  $\gamma_{2}$ & $0.0118$ \\
$X^1_0$& $\left(\begin{array}{cc}
2.3928  & 1.4489 \\ 
1.4489 & 2.2730
\end{array} \right)$  & & & & \\
\hline
Resnorm &$0.0034$&&&&\\
\hline 
\end{tabular} 
\caption{Calibration result for the Wishart-Gamma model. Resnorm represents the sum of squared distances between market and model implied volatilities.\label{table:Wishart-VGparams}}
\end{table}

\begin{figure}[ht]
  \centering
  \subfloat{\label{fig:priceErrorsWis\calibrationDate}\includegraphics[scale=0.40]{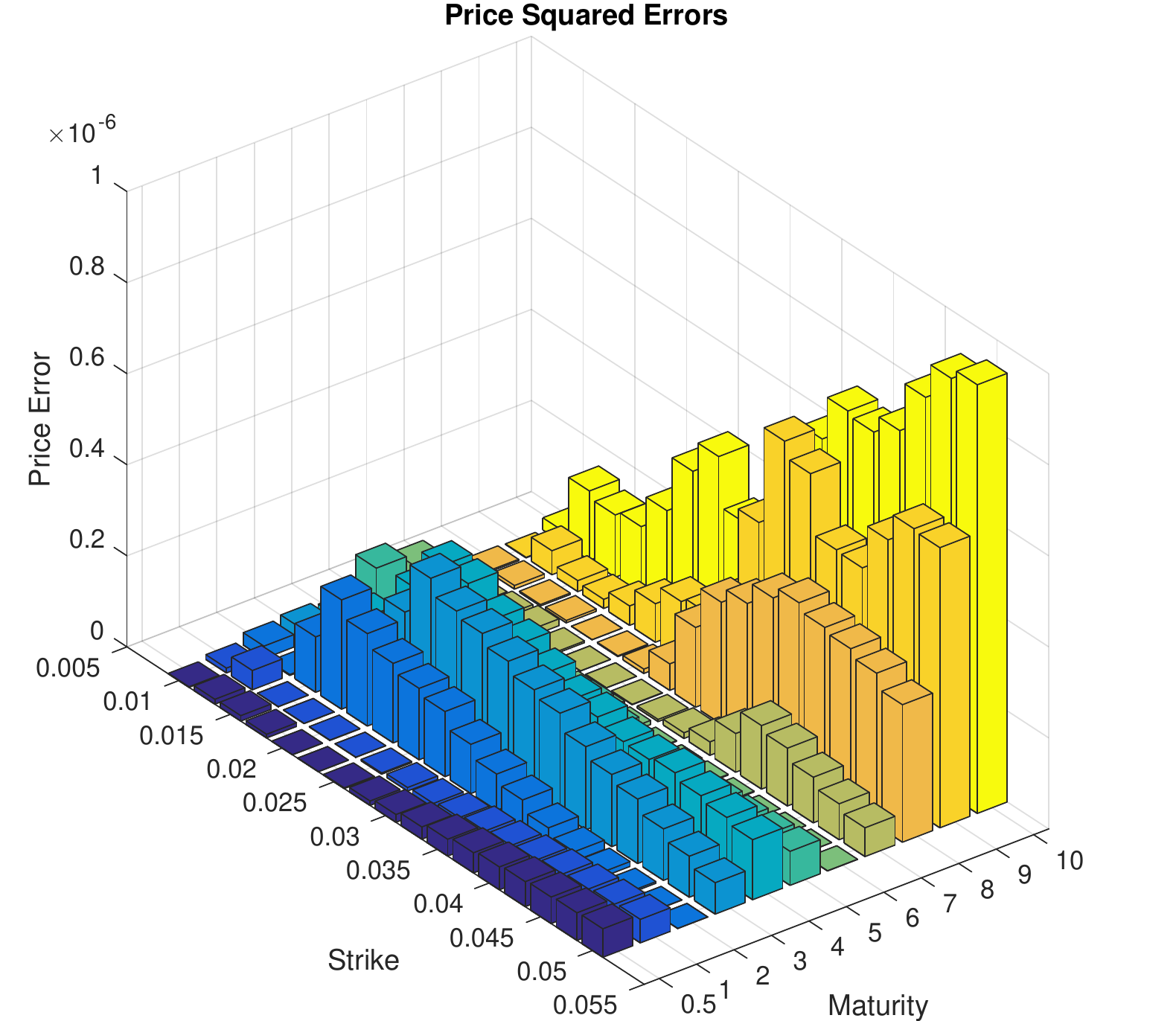}} \                                       \subfloat{\label{fig:volaErrorsWis\calibrationDate}\includegraphics[scale=0.40]{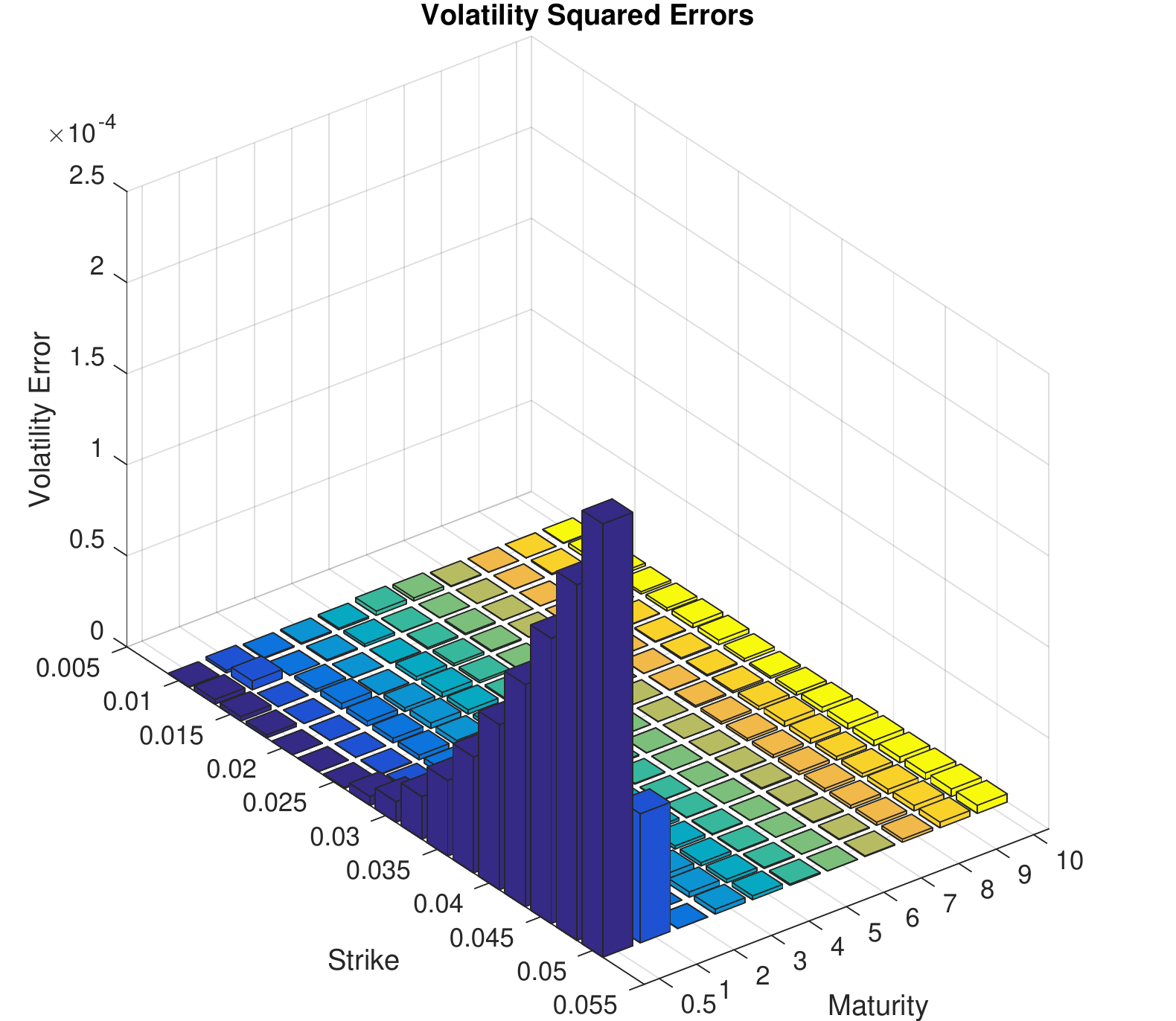}}\
\caption{Calibration residuals in terms of prices and implied volatilities. Calibration date: \theDate. \label{fig:calibrationErrorsWis\calibrationDate}}
\end{figure}

\subsubsection{Calibration stability} 
In this section, adopting the perspective of a front office interest rate desk performing a daily recalibration, we analyze the stability of the calibrated parameters. For this experiment we choose the CIR-Gamma model due to its higher computational tractability.

We take the calibrated parameters from Table~\ref{table:CIR-VGparams} as initial guess for a sequence of calibration experiments over a time window from August $2^{nd}$, 2011 up to August $31^{st}$, 2011.\footnote{Alternatively, one could choose as initial guess for the calibration at a given day the calibration result of the previous day: we performed this experiment and we noticed a slightly increased instability of the parameters. Therefore, our choice guarantees a higher stability while being consistent with the market practice.} 
 
The output of the procedure is a time series of calibrated parameters and calibration statistics. First of all, Figure~\ref{fig:calibrationStability06\calibrationDate} bottom panel shows that the quality of the fit, as measured by the sum of squared implied volatility errors, is subject to minor oscillations along the time window we consider, ranging between $0.0013$ and $0.0019$. A second important finding is related to the parameters $\gamma_1$ and $\gamma_2$ (central right panel): the ordering $\gamma_1<\gamma_2$  is persistent across the whole time window considered. 
We observe a satisfactory level of stability of the calibrated parameters (see top left and right panels) as measured by the ratio of the standard deviation rescaled by the sample average, which is always less than 20\%, see Table~\ref{table:CIRSTability1}. 

To further improve the stability of the calibration, one can proceed by noticing that scalar products among parameters, which are featured in the starting specification of the model, could generate instability in the calibration. Indeed, the specification of the risk-free short rate is proportional to the product between $\lambda$ and $X^1$, which suggests the presence of a redundancy between the projection and the initial value of the process $X^1_0$. Similarly, by looking at the spread specification, we also have products between $\gamma_i$, $i=1,2$, and the process $X^2$. For related issues when different affine processes generate the same term structure we refer to \cite{CFK:10}.
In the present case it is convenient to fix the value of $X^2_0$ and let the parameters $\gamma_i$ be calibrated, a choice which guarantees a better flexibility over the ordering of the multiplicative spreads associated to the different tenors. To test our intuition, we performed the stability experiment over the same time window we used previously. We were able to obtain a slight reduction of the coefficient of variation for almost all parameters without significantly affecting the quality of the fit in terms of mean squared error.

\begin{table}[!h]
\centering
\begin{tabular}{|c|c|c|c|}
\hline
Parameter & Std. Dev / Mean & Parameter & Std. Dev / Mean\\
\hline
\hline
$b$&$0.10532$ & $n$&$0.045478$ \\
$\beta$&$0.11868$ & $X^2_0$&$0.18247$ \\
$\sigma$& $0.094040$ & $\lambda$&$0.096756$ \\
$X^1_0$&$0.081350$ & $\gamma_1$&$0.069002$ \\
$m$&$0.10168$ & $\gamma_2$&$0.047381$  \\
\hline
\end{tabular}
\caption{Coefficient of variation for the model parameters.\label{table:CIRSTability1}}
\end{table}

\begin{figure}[h!]
  \centering
  \subfloat{\label{fig:calibrationStability01\calibrationDate}\includegraphics[scale=0.35]{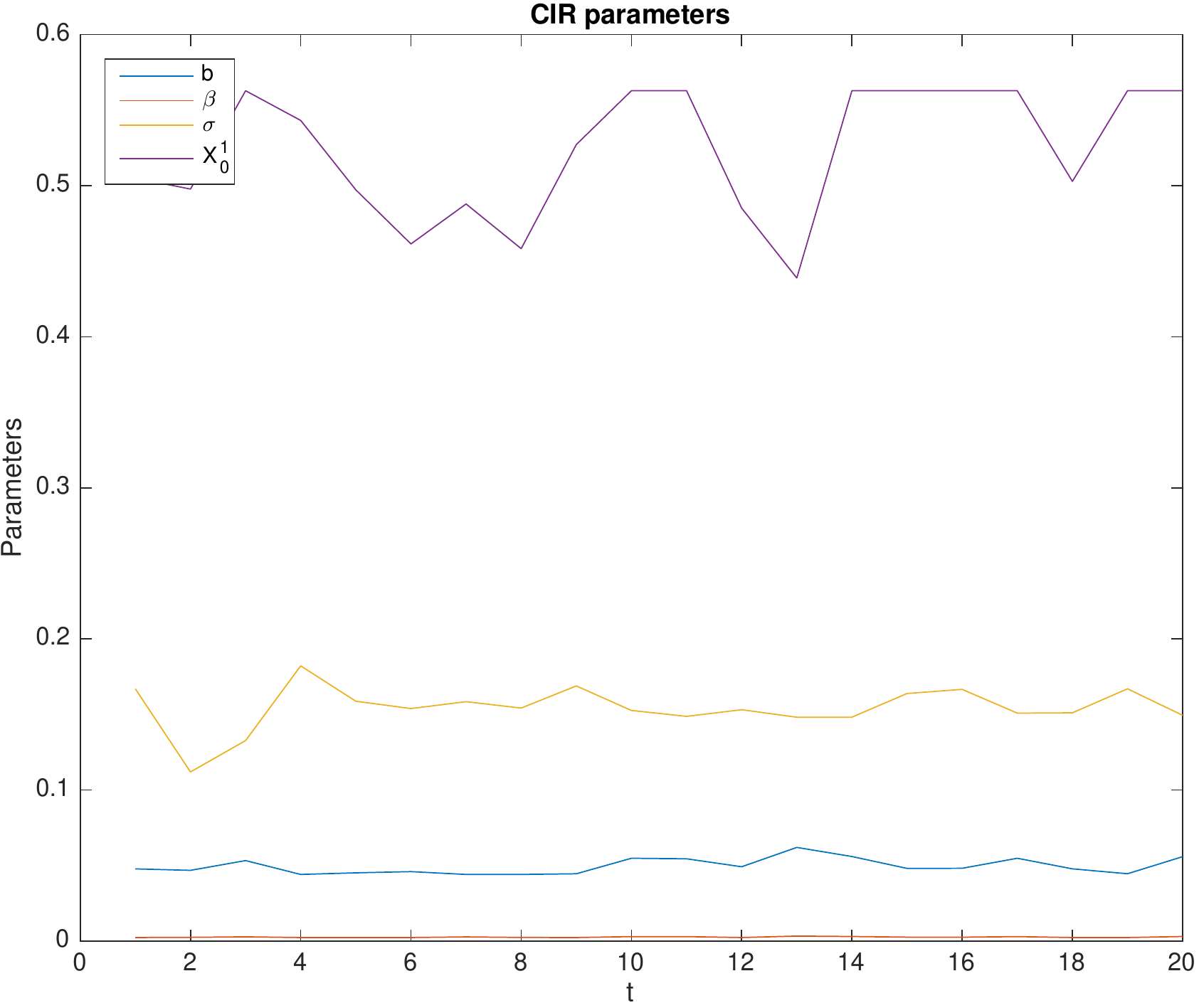}} \  
   \subfloat{\label{fig:calibrationStability02\calibrationDate}\includegraphics[scale=0.35]{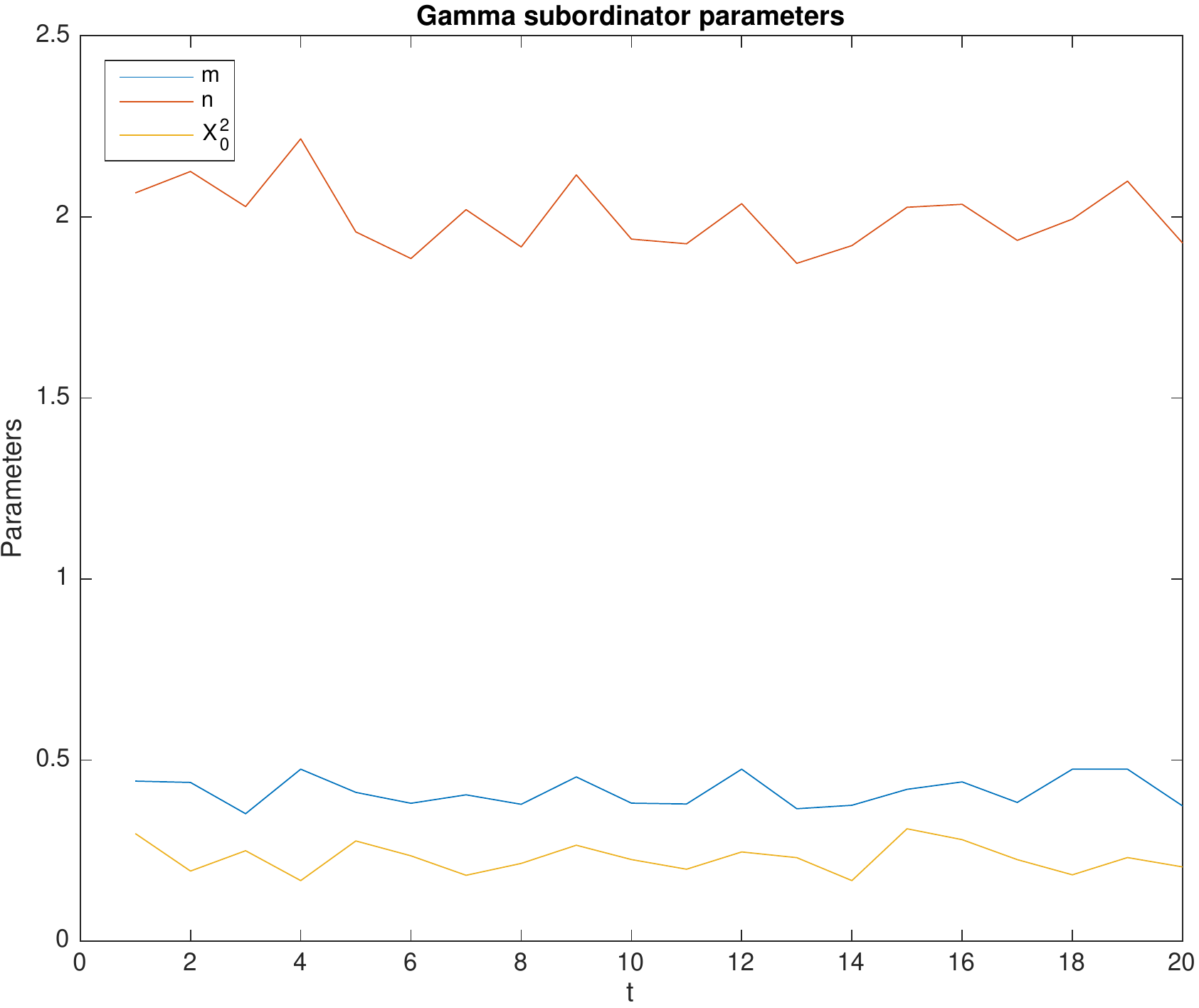}}\
    \subfloat{\label{fig:calibrationStability03\calibrationDate}\includegraphics[scale=0.35]{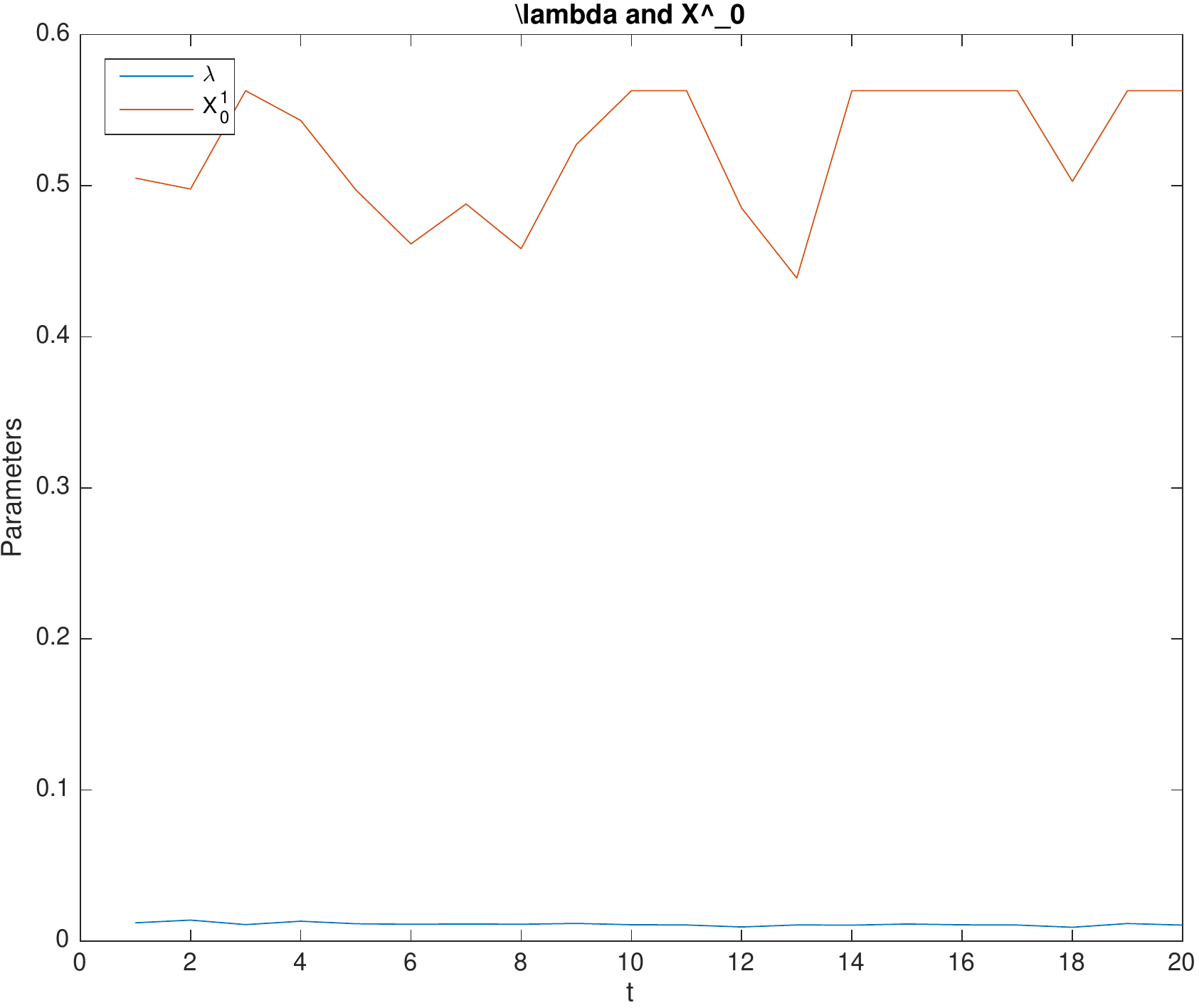}} \
    \subfloat{\label{fig:calibrationStability04\calibrationDate}\includegraphics[scale=0.35]{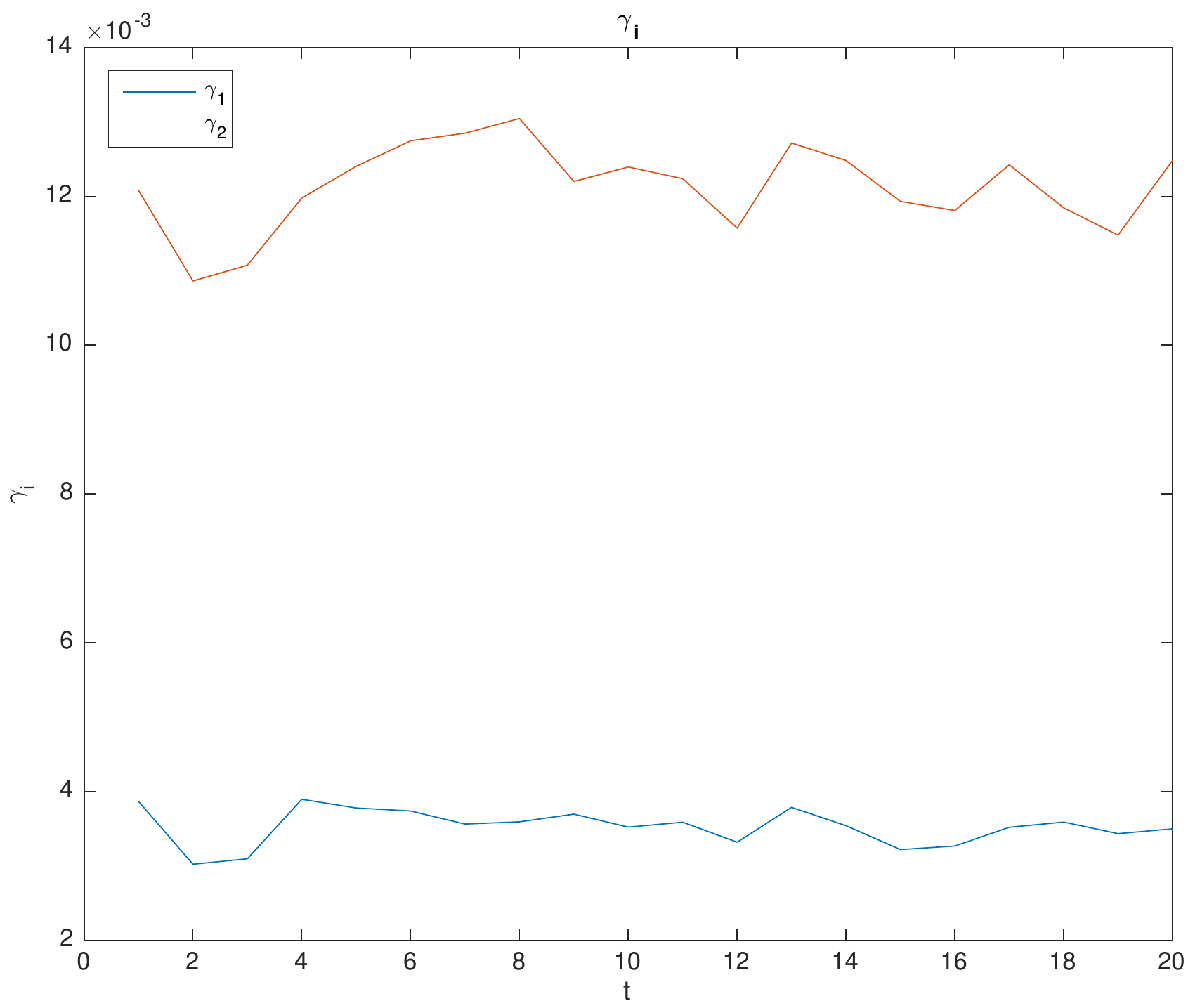}}   \
    \subfloat{\label{fig:calibrationStability05\calibrationDate}\includegraphics[scale=0.35]{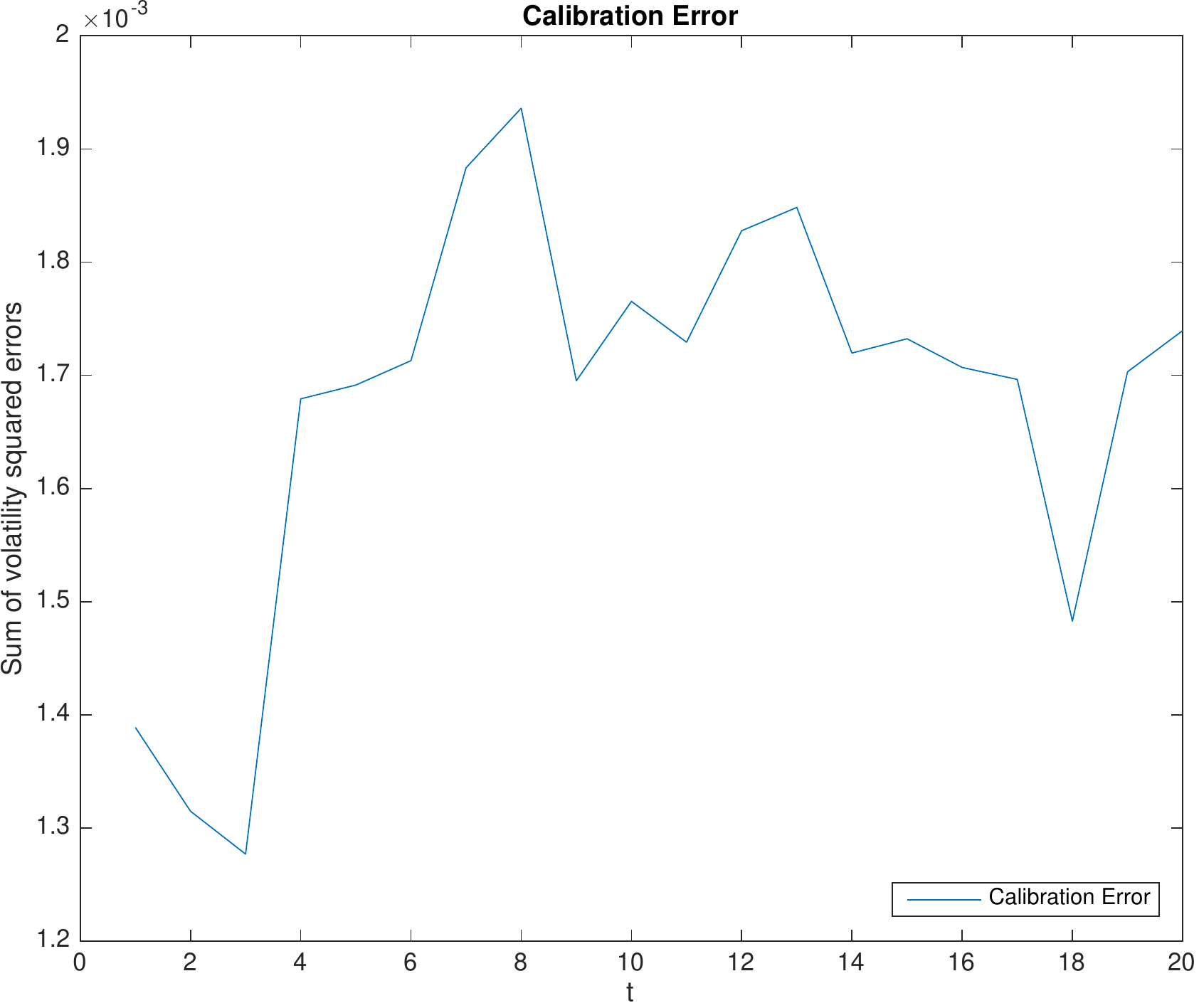}}                                        
\caption{Parameter stability test. Top left panel: CIR parameters. Top right panel: Gamma subordinator parameters. Central left panel: comparison between $X^1_0$ and $\lambda$. Central right panel: projection for the spread models.  Bottom panel: Sum of squared volatility errors over the whole sample. Calibration time window: August 2011. \label{fig:calibrationStability06\calibrationDate}}
\end{figure}

\appendix

\section{General pricing formulae}	\label{sec:gen_pricing}
This appendix presents general pricing formulas for typical interest rate derivatives. As we are going to show, the quantity $S^{\delta}(t,T)$ plays a pivotal role in the valuation of interest rate products. We here derive clean prices in the spirit of ~\cite[Appendix A]{CFG:14}, assuming perfect collateralization with the collateral asset being equal to the num\'eraire.

\subsection{Linear products}\label{sec:noopt}
The prices of linear interest rate products (i.e., without optionality features) can be directly expressed in terms of the basic quantities $B(t,T)$ and $S^{\delta}(t,T)$.\\

\noindent\emph{Forward rate agreement.}\\
A forward rate agreement (FRA) starting at $T$, with maturity $T+\delta$, fixed rate $K$ and notional $N$ is a contract which pays at time $T+\delta$ the amount
\begin{align*}
\Pi^{FRA}(T+\delta; T,T+\delta,K,N)=N\delta\bigl(L_T(T,T+\delta)-K\bigr).
\end{align*}
The value of such a claim at time $t\leq T$ is
\[
\begin{split}
\Pi^{FRA}(t; T,T+\delta,K,N)&=NB(t,T+\delta)\delta\,\Excond{\QQ^{T+\delta}}{L_T(T,T+\delta)-K}{\cF_t}\\
&=N\bigl(B(t,T)S^{\delta}(t,T)-B(t,T+\delta)(1+\delta K)\bigr).
\end{split}
\]

\noindent\emph{Overnight indexed swap.}\\ 
An overnight indexed swap (OIS) is a contract where two counterparties exchange two streams of payments: the first one is computed with respect to a fixed rate $K$, whereas the second one is indexed by an overnight rate (EONIA). Let us denote by $T_1, \ldots, T_n$ the payment dates, with $T_{i+1}-T_i=\delta$ for all $i=1,\ldots,n-1$. The swap is initiated at $T_0\in[0,T_1)$.  
The value of an OIS at $t\leq T_0$ with notional $N$ can be expressed as (see e.g.~\cite[Section 2.5]{fitr12})
\begin{align*}
\Pi^{OIS}(t; T_{1},T_{n},K,N)=N\left(B(t,T_0)-B(t,T_n)-K\delta\sum_{i=1}^{n} B(t,T_i)\right).
\end{align*}
Therefore, the OIS rate $K^{OIS}$, which is by definition the value for $K$ such that the OIS contract has zero value at inception, is given by
\[
K^{OIS}(T_1,T_n)=\frac{B(t,T_0)-B(t,T_n)}{\delta\sum_{k=1}^{n} B(t,T_k)}.
\]

\noindent\emph{Interest rate swap.}\\ 
In an interest rate swap (IRS), two streams of payments are exchanged between two counterparties: the first cash flow is computed with respect to a fixed rate $K$, whereas the second one is indexed by the  Libor rate. 
The value of the IRS at time $t\leq T_0$, where $T_0$ denotes the inception time, is given by
\begin{align*}
\Pi^{IRS}(t; T_{1},T_{n},K,N)%&=N\sum_{i=1}^{n}B(t,T_{i})\delta\,\Excond{\QQ^{T_i}}{L_{T_{i-1}}(T_{i-1},T_i)-K}{\cF_t}\\
%&=N\sum_{i=1}^{n}B(t,T_{i})\Excond{\QQ^{T_i}}{1+\delta L_{T_{i-1}}(T_{i-1},T_i)-(1+\delta K)}{\cF_t}\\
&=N\sum_{i=1}^{n}\left(B(t,T_{i-1})S^{\delta}(t,T_{i-1})-B(t,T_{i})(1+\delta K)\right).
\end{align*}
The swap rate $K^{IRS}$, which is by definition the value for $K$ such that the contract has zero value at inception, is given by
\begin{align*}
K^{IRS}(T_{1},T_{n},\delta)
= \frac{\sum_{i=1}^n\bigl(B(t,T_{i-1})S^{\delta}(t,T_{i-1})-B(t,T_i)\bigr)}{\delta\sum_{i=1}^nB(t,T_i)}
= \frac{\sum_{i=1}^nB(t,T_i)L_t(T_{i-1},T_1)}{\sum_{i=1}^nB(t,T_i)}.
\end{align*}

\noindent\emph{Basis swap.}\\ 
A basis swap is a special type of interest rate swap where two cash flows related to Libor rates associated to different tenors are exchanged between two counterparties. For instance, a typical basis swap may involve the exchange of the 3-month against the 6-month Libor rate. Following the standard convention for the definition of a basis swap in the Euro market (see~\cite{AB:13}), the basis swap is equivalent to a long/short position on two different interest rate swaps which share the same fixed leg. Let  $\cT^1=\left\{T^1_0,\cdots, T^1_{n_{1}}\right\}$, $\cT^2=\left\{T^2_0,\cdots T^2_{n_{2}}\right\}$ and $\cT^3=\left\{T^3_0,\cdots, T^3_{n_{3}}\right\}$, with $T^1_{n_{1}}=T^2_{n_{2}}=T^2_{n_{3}}$, $\cT^1\subset \cT^2$, $n_1<n_2$ and corresponding tenor lengths $\delta_1 > \delta_2$, with no constraints on $\delta_3$. The first two tenor structures on the one side and the third on the other are associated to the two floating and to the single fixed leg, respectively.  We denote by $N$ the notional of the swap, which is initiated at time $T^1_0=T^2_0=T^3_0$. The value at time $t\leq T^1_0$ is given by
\begin{align*}
\Pi^{BSW}(t;\cT^1,\cT^2,\cT^3,N)%\\
&=
N\left(\sum_{i=1}^{n_1}\bigl(B(t,T^1_{i-1})S^{\delta^1}(t,T^1_{i-1})-B(t,T^1_i)\bigr)\right.\\
&\quad-\sum_{j=1}^{n_2}\bigl(B(t,T^2_{j-1})S^{\delta^2}(t,T^2_{j-1})-B(t,T^2_j)\bigr)-\left.K\sum_{\ell=1}^{n_3}\delta^3B(t,T^3_{\ell})\right)\\
\end{align*}
The value $K^{BSW}$ (called \textit{basis swap spread}) such that the value of the contract at initiation is zero is then given by
\begin{align*}
K^{BSW}(\cT^1,\cT^2,\cT^3)&=
\frac{\sum_{i=1}^{n_1}\bigl(B(t,T^1_{i-1})S^{\delta^1}(t,T^1_{i-1})-B(t,T^1_i)\bigr)-\sum_{j=1}^{n_2}\bigl(B(t,T^2_{j-1})S^{\delta^2}(t,T^2_{j-1})-B(t,T^2_j)\bigr)}{\delta_3\sum_{\ell=1}^{n_3}B(t,T^3_{\ell})}.
\end{align*}
Note that, prior to the financial crisis, the value of $K^{BSW}$ used to be (approximately) zero.

\subsection{Products with optionality features}	\label{sec:option}
In this section,  we report general valuation formulas  for plain vanilla interest rate products such as (European) caplets and swaptions. 
\noindent\emph{Caplet.}\\ 
The price at time $t$ of a caplet with strike price $K$, maturity $T$, settled in arrears at $T+\delta$, is given by
\begin{align}
\Pi^{CPLT}(t; T,T+\delta,K,N)&=NB_t\delta\,\Excond{}{\frac{1}{B_{T+\delta}}\Bigl(L_T(T,T+\delta)-K\Bigr)^+}{\cF_t}\notag\\
&=N\Excond{}{\frac{B_t}{B_T}\Bigl(S^{\delta}(T,T)-(1+\delta K)B(T,T+\delta)\Bigr)^+}{\cF_t}	\label{eq:caplet_gen}.
\end{align}

\begin{remark}
Note that, in the classical single-curve setting (i.e., under the assumption that $S^{\delta}(T,T)$ is identically equal to one), the valuation formula \eqref{eq:caplet_gen} reduces to the classical relationship between a caplet and a put option on a zero-coupon bond with strike $1/(1+\delta K)$. 
\end{remark}

From~\eqref{eq:caplet_gen}, we see that the payoff of a caplet at time $T$ and notional $N=1$ corresponds to
\begin{align*}
\Bigl(S^{\delta}(T,T)-(1+\delta K)B(T,T+\delta)\Bigr)^+&=S^{\delta}(T,T)1_{\{S^{\delta}(T,T)\geq (1+\delta K)B(T,T+\delta)\}}\\
&\quad -(1+\delta K)B(T,T+\delta)1_{\{S^{\delta}(T,T)\geq (1+\delta K)B(T,T+\delta)\}}.
\end{align*}
Let us now introduce the following probability measure on $\mathcal{F}_T$:
\[
\frac{d\widetilde{\mathbb{Q}}}{d \mathbb{Q}}:=\frac{S^{\delta}(T,T) B(T,T)}{B_T S^{\delta}(0,T) B(0,T)}.
\]
Note that $\frac{d\widetilde{\mathbb{Q}}}{d \mathbb{Q}} >0$ and 
$
\mathbb{E}^{\mathbb{Q}}\left[\frac{S^{\delta}(T,T) B(T,T)}{B_T S^{\delta}(0,T) B(0,T)}\right]=1
$,
since $\frac{S^{\delta}(t,T) B(t,T)}{B_t B(0,T)}$ is a $\mathbb{Q}$-martingale by Proposition~\ref{prop:bonds_spreads}. By Bayes' formula, we therefore have
\[
\mathbb{E}^{\mathbb{Q}}\left[\frac{B_t}{B_T}S^{\delta}(T,T)1_{\{S^{\delta}(T,T)\geq (1+\delta K)B(T,T+\delta)\}}\, \Bigr|\,  \mathcal{F}_t\right]=S^{\delta}(t,T)B(t,T)\mathbb{E}^{\widetilde{\mathbb{Q}}}\left[1_{\{S^{\delta}(T,T)\geq (1+\delta K)B(T,T+\delta)\}}\, \Bigr| \, \mathcal{F}_t\right].
\]
Similarly, we obtain 
\begin{align*}
\mathbb{E}^{\mathbb{Q}}\left[\frac{B_t}{B_T}(1+\delta K)B(T,T+\delta)1_{\{S^{\delta}(T,T)\geq (1+\delta K)B(T,T+\delta)\}}\, \Bigr|\,  \mathcal{F}_t\right]\\=(1+\delta K)B(t,T+\delta)\mathbb{E}^{\mathbb{Q}^{T+\delta}}\left[1_{\{S^{\delta}(T,T)\geq (1+\delta K)B(T,T+\delta)\}}\, \Bigr|\,  \mathcal{F}_t\right].
\end{align*}
Hence, the price of a caplet can be computed via
\begin{equation}\label{eq:capletpriceWish}
\begin{split}
\Pi^{CPLT}(t; T,T+\delta,K,1)&=S^{\delta}(t,T)B(t,T)\widetilde{\mathbb{Q}}\left[S^{\delta}(T,T)\geq (1+\delta K)B(T,T+\delta)\, \Bigr| \, \mathcal{F}_t\right]\\
&\quad -(1+\delta K)B(t,T+\delta)\mathbb{Q}^{T+\delta}\left[S^{\delta}(T,T)\geq (1+\delta K)B(T,T+\delta)\, \Bigr|\,  \mathcal{F}_t\right].
\end{split}
\end{equation}

\noindent\emph{Swaption.}\\ 
We consider a standard European payer swaption with maturity $T$, written on a (payer) interest rate swap starting at $T_0=T$ and payment dates $T_1,..., T_n$, with $T_{i+1}-T_i=\delta$ for all $i=1,\ldots,n-1$, with notional $N$. The value of such a claim at time $t$ is given by

\[
\Pi^{SWPTN}(t; T_{1},T_{n},K,N)
=N\mathbb{E}^{}\left[\frac{B_t}{B_T}\left(\sum_{i=1}^{n}B(T,T_{i-1})S^{\delta}(T,T_{i-1})-(1+\delta K)B(T,T_i)\right)^+\biggr|\cF_t\right].
\]

\section{Proof of Proposition~\ref{prop:bonds_spreads}}	\label{app:proof_prop}

Under Assumption \ref{ass:couple}, the process $(B(t,T)/B_t)_{0\leq t\leq T}$ is a martingale, for all $T\in[0,\T]$. Since $B(T,T)=1$, this implies that
\[
\frac{B(t,T)}{B_t} = \EE\left[\frac{B(T,T)}{B_T}\Bigr|\cF_t\right]
= \EE\left[\frac{1}{B_T}\Bigr|\cF_t\right],
\qquad\text{ for all }0\leq t\leq T\leq \T,
\]
thus proving part (i).
In particular, note that this implies that $B(0,T)=\EE[1/B_T]$, thus ensuring that  $d\QQ^T/d\QQ = 1/(B_TB(0,T))$ defines a probability measure $\QQ^T\sim\QQ$, for every $T\in[0,\T]$.
Recalling that $\Pi^{FRA}(t;T,T+\delta_i,L_t(T,T+\delta_i))=0$, for all $0\leq t\leq T\leq \T$ and $i=1,\ldots,m$, it holds that
\begin{align*}
0 &= \frac{1}{\delta_i}\frac{\Pi^{FRA}(t;T,T+\delta_i,L_t(T,T+\delta_i))}{B_t}
= \frac{1}{\delta_i}\EE\left[\frac{\Pi^{FRA}(T+\delta_i;T,T+\delta_i,L_t(T,T+\delta_i))}{B_{T+\delta_i}}\Bigr|\cF_t\right]	\\
&= \EE\left[\frac{L_T(T,T+\delta_i)-L_t(T,T+\delta_i)}{B_{T+\delta_i}}\Bigr|\cF_t\right]
= \frac{B(t,T+\delta_i)}{B_t}\bigl(\EE^{\QQ^{T+\delta_i}}[L_T(T,T+\delta_i)|\cF_t]-L_t(T,T+\delta_i)\bigr).
\end{align*}
Since $B(t,T+\delta_i)/B_t>0$ by part (i), the last equality proves part (ii).
Finally, by Bayes' formula, the process $(S^{\delta_i}(t,T))_{0\leq t\leq T}$ is a $\QQ^T$-martingale if and only if the process $(M^i_t)_{0\leq t\leq T}$ defined by
\[
M^i_t := S^{\delta_i}(t,T)\frac{d\QQ^T|_{\cF_t}}{d\QQ^{T+\delta_i}|_{\cF_t}}
= S^{\delta_i}(t,T)\frac{B(t,T)B(0,T+\delta_i)}{B(t,T+{\delta_i})B(0,T)}
\]
is a $\QQ^{T+\delta_i}$-martingale, for every $i=1,\ldots,m$. By definition of $S^{\delta_i}(t,T)$, it holds that
\[
M^i_t = \frac{1+\delta_i L_t(T,T+\delta_i)}{1+\delta_i \Lois_t(T,T+\delta_i)}\frac{B(t,T)B(0,T+\delta_i)}{B(t,T+{\delta_i})B(0,T)}
= \bigl(1+\delta_i L_t(T,T+\delta_i)\bigr)\frac{B(0,T+\delta_i)}{B(0,T)}
\]
and the desired martingale property then follows from part (ii).

%\bibliographystyle{abbrv}

%\bibliography{biblio160220}

\end{document}